\documentclass[final,12pt]{colt2025}
\pdfoutput=1

\usepackage{times} 
\usepackage{xcolor}
\usepackage{caption}
\usepackage{booktabs} 

\usepackage{amsmath, mathrsfs, amsfonts, nicefrac}
\usepackage{thmtools,thm-restate} 
\usepackage{aligned-overset}
\usepackage{enumitem}
\allowdisplaybreaks[2]

\newenvironment{proofsketch}{\paragraph{Proof sketch}}{\hfill$\blacksquare$}

\SetKwInput{Init}{initialization}
\SetKwInput{Input}{input}

\usepackage[capitalize,noabbrev,nameinlink]{cleveref}

\crefalias{learningprotocol}{algocf}
\crefname{learningprotocol}{learning protocol}{learning protocols}
\Crefname{learningprotocol}{Learning protocol}{Learning protocols}

\crefalias{tradingprotocol}{algocf}
\crefname{tradingprotocol}{trading protocol}{trading protocols}
\Crefname{tradingprotocol}{Trading Protocol}{Trading protocols}

\newcommand{\cA}{\mathcal{A}}
\newcommand{\cB}{\mathcal{B}}

\newcommand{\cD}{\mathcal{D}}

\newcommand{\cF}{\mathcal{F}}

\newcommand{\cO}{\mathcal{O}}
\newcommand{\cP}{\mathcal{P}}

\newcommand{\cR}{\mathcal{R}}
\newcommand{\cS}{\mathcal{S}}

\newcommand{\cU}{\mathcal{U}}

\newcommand{\cX}{\mathcal{X}}

\newcommand{\E}{\mathbb{E}}

\newcommand{\I}{\mathbb{I}}

\newcommand{\N}{\mathbb{N}}

\newcommand{\Pb}{\mathbb{P}}
\newcommand{\bbR}{\mathbb{R}}

\newcommand{\e}{\varepsilon}

\newcommand{\lrb}[1]{\left(#1\right)}
\newcommand{\brb}[1]{\bigl(#1\bigr)}
\newcommand{\Brb}[1]{\Bigl(#1\Bigr)}

\newcommand{\lsb}[1]{\left[#1\right]}
\newcommand{\bsb}[1]{\bigl[#1\bigr]}
\newcommand{\Bsb}[1]{\Bigl[#1\Bigr]}

\newcommand{\lcb}[1]{\left\{#1\right\}}
\newcommand{\bcb}[1]{\bigl\{#1\bigr\}}

\newcommand{\one}[1]{\mathbb{I}\{#1\}}

\newcommand{\labs}[1]{\left\lvert#1\right\rvert}

\newcommand{\Babs}[1]{\Bigl\lvert#1\Bigr\rvert}

\newcommand{\lno}[1]{\left\lVert#1\right\rVert}

\DeclareMathOperator*{\argmin}{argmin}

\newcommand{\dif}{\,\mathrm{d}}

\newcommand{\nhphantom}[1]{\sbox0{#1}\hspace{-\the\wd0}}

\newcommand{\ceq}{\coloneqq}

\newcommand{\mypapertitle}{Market Making without Regret}
\newcommand{\util}{u}
\newcommand{\Util}{U}
\newcommand{\vin}{\texttt{M3}}
\newcommand{\db}{\texttt{Discretized Bandits for First Price Auctions}}
\newcommand{\dbdp}{\texttt{Discretized Bandits for Dynamic Pricing}}
\newcommand{\itemm}{\item[$\bullet$]}
\newcommand{\Rfpa}{R^{\mathrm{fpa}}}

\newcommand{\Rdp}{R^{\mathrm{dp}}}

\newcommand{\ftbp}{\texttt{FTABP}}

\newcommand{\nP}{M}
\newcommand{\np}{m}
\newcommand{\nZ}{V}
\newcommand{\nz}{v}
\newcommand{\nB}{B}
\newcommand{\ncB}{\cB} 
\newcommand{\nb}{b}
\newcommand{\nS}{A}
\newcommand{\ncS}{\cS}
\newcommand{\ns}{a}
\newcommand{\nX}{X}
\newcommand{\nx}{x}

\newcommand{\nV}{Z}

\newcommand{\nM}{H}

\newcommand{\nC}{C}
\newcommand{\nc}{c}
\newcommand{\nW}{W}
\newcommand{\nw}{w}

\newcommand{\nY}{P}
\newcommand{\ny}{p}

\renewcommand{\tilde}{\widetilde}

\newcommand{\RK}{3}
\newcommand{\TK}{3}
\newcommand{\Tk}{2}

\title{\mypapertitle}
\date{}

\coltauthor{%
    \Name{Nicol\`o {Cesa-Bianchi}} \Email{nicolo.cesa-bianchi@unimi.it}\\
    \addr{Universit\`a degli Studi di Milano and Politecnico di Milano, Milan, Italy}
    \AND%
    \Name{Tommaso Cesari} \Email{tcesari@uottowa.ca}\\
    \addr{University of Ottawa, Ottawa, Canada}
    \AND%
    \Name{Roberto Colomboni} \Email{roberto.colomboni@polimi.it}\\
    \addr{Politecnico di Milano and Universit\`a degli Studi di Milano, Milan, Italy}
    \AND%
    \Name{Luigi Foscari} \Email{luigi.foscari@unimi.it}\\
    \addr{Universit\`a degli Studi di Milano, Milan, Italy}
    \AND%
    \Name{Vinayak Pathak} \Email{path.vinayak@gmail.com}\\
    \addr{Independent researcher}
}

\begin{document}

\maketitle 


\begin{abstract}
We consider a sequential decision-making setting where, at every round $t$, the learner (a \emph{market maker}) posts a \emph{bid} price $\nB_t$ and an \emph{ask} price $\nS_t$ to an incoming trader (the \emph{taker}) with a private valuation for some asset. 
If the trader's valuation is lower than the bid price, or higher than the ask price, then a trade (sell or buy) occurs. 
Letting $\nP_t$ be the market price (observed only at the end of round $t$), the maker's utility is $\nP_t-\nB_t$ if the maker bought the asset, it is $\nS_t-\nP_t$ if they sold it, and it is $0$ if no trade occurred. 
We characterize the maker's regret with respect to the best fixed choice of bid and ask pairs under a variety of assumptions (adversarial, i.i.d., and their variants) on the sequence of market prices and valuations. 
Our upper bound analysis unveils an intriguing connection relating market making to first-price auctions and dynamic pricing. 
Our main technical contribution is a lower bound for the i.i.d.\ case with Lipschitz distributions and independence between market prices and takers' valuations.
The difficulty in the analysis stems from a unique relationship between the reward and feedback functions that allows learning algorithms to trade off reward for information in a continuous way.
\end{abstract}

\begin{keywords}
Regret minimization, online learning, market making, first-price auctions, dynamic pricing.
\end{keywords}




\section{Introduction}
Trading in financial markets is a crucial activity that helps keep the world's economy running, and several players, including hedge funds, prop trading firms, investment banks, central banks, and retail traders participate in it daily. While every actor has their own objective function (for example, a hedge fund wants to maximize profit whereas a central bank wants to keep inflation in check), at a fundamental level, trading can be viewed as a stochastic control problem where agents want the state to evolve so as to maximize their objective function.
In this work we focus on market makers, i.e., traders whose job is to facilitate other trades to happen.
One way to do so is by broadcasting, at all times, a price (\emph{bid}) at which they are willing to buy, and a price (\emph{ask}) at which they are willing to sell the asset being traded.
This way when a buyer (resp., seller) arrives, they do not have to wait for a seller (resp., buyer) to be able to perform a transaction.
A market where it is easy to make trades is called \emph{liquid}, and liquidity is a desirable property for any kind of market. A market maker thus provides an essential service by increasing the liquidity of the market. 
Market making is challenging, and a lot of thought goes into making it profitable, see~\citet{harris2003trading} for an overview. A major risk that a market maker has to deal with, called adverse selection, is the risk that your counterpart is an \emph{informed} trader who knows something about the future direction of the price movement. 
For example, an informed trader who knows that an asset is soon about to become cheaper will sell it to you thus forcing you to buy something whose price will crash. 
One way to mitigate this risk is to immediately offload your positions elsewhere. More generally, a market maker does not specialize in predicting the future movements of prices, and thus dislikes holding assets for too long. Market making desks often have limits on how much inventory they are allowed to hold at a time \citep{fender2015early}.


The strategy of immediately offloading your position, though simple, is profitable if one does it strategically. 
Consider, for example, the case where one can trade the same asset on two exchanges, where its bid-ask spread on the first is smaller than its bid-ask spread on the second. 
In this case, one can make a market on the second while offloading their positions on the first. 
This situation arises frequently for two closely related assets.
For example, a \emph{futures} contract on an asset is often more liquid than the asset itself and their price movements are highly correlated. 
For these reasons, we focus on a setting where the market maker offloads their positions shortly after each transaction. 

\subsection{Related works}

\paragraph{Online market making.}
Our work is closely related to that of \citet{abernethy2013adaptive}, which analyzes competition against $N$ constant-spread dynamic strategies. While we also consider constant-spread (static) strategies, key differences make their results non-comparable to ours. Their model assumes full information feedback, with an adversarial yet bounded market price change and a fully known utility function at each step, leading to a regret of $\sqrt{T\ln N}$. In contrast, we compete against a continuum of ask-bid strategies while managing both estimation and approximation errors. Crucially, our model lacks full information feedback, as rewards depend on private valuations, requiring structural exploitation of the reward function to mitigate uncertainty.

Our work is also related to the classic Glosten-Milgrom model \citep{glostenBidAskTransaction1985}, which introduced the concept of a market \textit{specialist} (maker) providing liquidity to the market within an exchange. 
The specialist interacts with both informed traders (whose valuations are influenced by future market prices) and uninformed ones.
The specialist's objective is to distinguish between informed and uninformed traders via interactions in an online fashion and optimize the bid-ask spread accordingly. 
In \citeyear{learnigmarketmakerGlostenMilgrom2005}, \citeauthor{learnigmarketmakerGlostenMilgrom2005} considers a learning algorithm in an extended version of the Glosten-Milgrom model for the market maker with a third type of ``noisy-informed'' trader, whose current valuation is a noisy version of the future market price. However, to the best of our knowledge, they do not provide any regret guarantees.

\paragraph{Smoothed adversary.}
Popularized by \citet{spielman2004smoothed,rakhlin2011online,haghtalab2020smoothed}, smoothness analysis provides a framework for analyzing algorithms in problems parameterized by distributions that are not ``too'' concentrated. 
Recent advancements in the smoothed analysis of online learning algorithms include contributions by \citet{kannan2018smoothed,haghtalaboracle,HaghtalabRS21,block2022smoothed,DurvasulaHZ23,cesa2021regret,cesa2023repeated,bolic2024online,JMLR:v25:23-1627}.
In this work, we exploit the connection between the smoothness of the takers' valuations distributions (i.e., the Lipschitzness of the corresponding cumulative distribution function, or, equivalently, the boundedness of the corresponding probability density function) and Lipschitzness of the expected utility. 
This property is crucial to achieving sublinear regret guarantees in adversarial settings and
affects learning depending on different assumptions on the level of information shared with the learner while bidding.

Please see \Cref{sec:additional} for further related works.

\paragraph{Online learning in digital markets.}
Our work unveils a connection between market making and first-price auctions/dynamic pricing, thus contributing to the existing body of research on
the multi-armed bandit framework, which boasts a substantial literature base across statistics, operations research, computer science, and economics \citep{auerFinitetimeAnalysisMultiarmed2002,lattimore2020bandit,badanidiyuru_bandits_2018,Kleinberg2004NearlyTB}.
The multi-armed bandit framework has proven to be an effective tool for addressing key problems in finance under uncertainty, including repeated first-price auctions \citep{cesabianchiRoleTransparencyRepeated2024,han2024optimalnoregretlearningrepeated,pmlr-v157-achddou21a,NIPS2017_0bed45bd} and dynamic pricing \citep{cesa2019dynamic,babaioffDynamicPricingLimited2015,kleinberg2003value}. These settings are closely related to online market making (see \Cref{sec:additional} for further discussion).

Closely related to dynamic pricing is the problem of optimal taxation \citep{taxation2025}, where the interaction between the learner and the environment follows a similar structure, but the objective shifts from revenue maximization to the design of pricing mechanisms that maximize social welfare.

Beyond first-price auction settings, related auction formats include repeated Vickrey auctions, studied from both the bidder’s perspective \citep{weedOnlineLearningRepeated2015} and the seller’s perspective \citep{cesaSecondPrice}.

Another important problem at the intersection of online learning and digital markets is repeated bilateral trade.
In this setting, the learner acts as a broker, proposing price pairs to a sequence of traders with unknown private valuations.
The foundational works of \citet{cesa2021regret,cesa2024bilateral} established optimal regret rates under i.i.d.\ valuations, assuming smoothness and independence between buyer and seller values, using the gain from trade as reward function, and mechanisms that are strongly budget balanced.
This line of work has been extended in multiple directions: \citet{azar2022alpha} proposed weakly budget-balanced mechanisms achieving optimal $2$-regret in adversarial settings; \citet{cesa2023repeated,JMLR:v25:23-1627} relaxed the i.i.d.\ and the independence between sellers and buyers assumptions while maintaining optimal regret guarantees under smoothness alone; and \citet{bernasconi2024no} introduced globally budget-balanced mechanisms with sublinear regret even in adversarial settings.

Contextual versions of bilateral trade have also been studied by \citet{GaucherBCCP25}, while \citet{fairgainfromtrade} took a different perspective by proposing a fairness-oriented reward function establishing optimal regret bounds under the i.i.d.\ assumption with independent buyer and seller valuations, using strongly budget-balanced mechanisms.

Recent work has also considered scenarios in which agents are not assigned fixed roles as buyers or sellers. In this more general setting, \citet{bolic2024online} achieved faster regret rates than in standard bilateral trade using strongly budget-balanced mechanisms under the same assumptions as \citet{cesa2021regret}. In the contextual setting, \citet{bachoc2025contextualLinear} obtained optimal fast rates under a noisy-linear model, and \citet{bachoc2025contextual} extended these results to nonparametric contexts. Finally, \citet{cesari2025volumeFairnesss} studied the objective of maximizing the volume of trade, achieving accelerated regret rates in this setting.

\subsection{Our contributions}
We study the question of market making in an online learning setup.
Here, trading happens in discrete time steps and an unknown stochastic process governs the prices of the asset being traded as well as the private valuations assigned to the asset by market participants (see details in~\cref{sec:setting}).
The market maker is an online learning algorithm that posts a bid and an ask at the beginning of each time step and receives some feedback at the end of the time step. 
We are interested in controlling the regret suffered by this learner at the end of $T$ time steps.
At the end of each round, the learner observes the market price, i.e., the price at which they are able to offload their position. Additionally, they see whether their bid or their ask (or neither) got successfully traded.
We consider settings where the unknown stochastic process is independent and identically distributed (iid), as well as the adversarial setting (adv), where no such assumption is made on the process. We also consider settings where the cumulative distribution functions of takers' valuations are Lipschitz (lip)
, as well as settings where at each time step $t$ the corresponding market price and taker's valuation are independent of each other (iv).
Our results are summarized in \Cref{tab:our-results}.
\begin{table}
    \centering
    \begin{tabular}{l|c c|c c c c}
        \toprule
        Setting &  (adv) &  (adv+lip) &  (iid) &  (iid+lip) &  (iid+iv) & (iid+lip+iv) \\
        \midrule
        Regret & $T$ & $T^{\nicefrac{2}{3}}$ & $T$ &  $T^{\nicefrac{2}{3}}$ & $T^{\nicefrac{2}{3}}$ & $T^{\nicefrac{2}{3}}$ \\
        \bottomrule
    \end{tabular}
    \caption{A summary of the regret guarantees we prove for each variant of the online market-making problem. All the rates are optimal up to constant factors.\label{tab:our-results}}
\end{table}

We make the following contributions.
\begin{enumerate}[noitemsep]
    \item We design the \vin{} (meta-)algorithm (\Cref{a:meta-vin}) and prove a $\cO(T^{\nicefrac{2}{3}})$ upper bound on its regret under the assumption that either the cumulative distribution functions of the takers' valuations are Lipschitz (lip) or the sequence of market prices and takers' valuations are independent and identically distributed (iid) with market prices being independent of takers' valuations (iv) (\Cref{t:upper-bound-vin}).
    \item Our main technical contribution is a lower bound of order $\Omega(T^{\nicefrac{2}{3}})$ that matches \vin{}'s upper bound, and holds even under the simultaneous assumptions that the sequence of market prices and takers' valuations is independent and identically distributed (iid), market prices are independent of takers' valuations (iv), the cumulative distribution function of the takers' valuations is Lipschitz (lip), and the cumulative distribution function of the market prices is Lipschitz (\Cref{t:lower-bound-lip-iv}).
    \item We then investigate the necessity of assuming either Lipschitzness of the cumulative distribution function of the takers' valuations (lip) or the independence of market prices and takers' valuations (iv). We prove that, if both assumptions are dropped, learning is impossible in general, even when market prices and takers' valuations are independent and identically distributed (iid)\ (\Cref{t:minimal-lower-bound-iid}).
\end{enumerate}

We can interpret the independence between takers' valuation and market prices (iv) as the taker being an uninformed trader (in the same spirit as the Glosten-Milgrom model). If we consider the (iid) setting, and contrast the impossibility of learning when neither Lipschitzness (lip) nor independence (iv) is assumed with the $\cO(T^{\nicefrac{2}{3}})$ upper bound when the market prices are independent of takers' valuations (iv), we get evidence for the practitioner's intuition that the existence of uninformed traders is good for the market maker.

\section{Setting}\label{sec:setting}
\begin{algorithm2e}[t]
    \renewcommand*{\algorithmcfname}{Trading protocol}

    \For {time $t=1,2,\ldots$
    }{
        Taker arrives with a private valuation $\nZ_t \in [0,1]$\;
        Maker posts bid/ask prices $(\nB_t , \nS_t) \in \cU$\;
        Feedback $\one{ \nB_t \ge \nZ_t }$ and $\one{ \nS_t < \nZ_t }$ is revealed\;
        \lIf{$\nB_t \ge \nZ_t$}{%
            Maker buys and pays $\nB_t$ to Taker%
        } \lElseIf{$\nS_t < \nZ_t$}{%
            Maker sells (short)\hyperlink{footnote1}{\textsuperscript{1}} and is paid $\nS_t$ by Taker%
        } \lElse {no trade happens}
        The market price $\nP_t\in[0,1]$ is revealed\;

        \lIf{Maker bought from Taker}{%
            Maker offloads to Market at price $\nP_t$%
        }\lElseIf{Maker sold (short) to Taker}{%
            Maker buys from Market at price $\nP_t$%
        }
        Maker's utility is $\Util_t(\nB_t,\nS_t) \ceq \util(\nB_t,\nS_t,\nP_t,\nZ_t)$\;
    }
    \caption{Online market making}
    \label[tradingprotocol]{a:trading-protocol}
\end{algorithm2e}
\hypertarget{footnote1}{\footnotetext[1]{We remark that in most markets, traders can sell assets they do not currently own (\emph{short-selling}; see, e.g., the classic \citealt{black1973pricing}). For this reason, we do not assume that the market maker owns a unit of the asset before selling it. If the market maker sells short, they will buy the asset at the market price at the end of the interaction.}}%
\setcounter{footnote}{1}%
In online market making, the action space of the maker is the upper triangle $\cU \coloneqq \bcb{ (\nb, \ns) \in [0,1]^2 \,:\, \nb \le \ns } $, enforcing the constraint that a bid price $\nb$ is never larger than the corresponding ask price $\ns$.
For an overview of the notation, see \Cref{tab:notation}, in \Cref{s:appe-notation}.
The utility of the market maker, for all $(\nb,\ns) \in \cU$ and $\np,\nz\in[0,1]$, is
\begin{equation}
\label{e:utility-of-market-maker}    
    \util(\nb,\ns,\np,\nz) 
\ceq
    \underbrace{ (\np-\nb) }_{\text{market-bid spread}} \!\!\!\!\times\; \underbrace{ \one{ \nb \ge \nz } }_{\text{Maker buys}}
    \;\;\; +
    \underbrace{ (\ns-\np) }_{\text{ask-market spread}} \!\!\!\!\times\; \underbrace{ \one{ \ns < \nz } }_{\text{Maker sells}}
    \; ,
\end{equation}
where $\nz$ is the taker's private valuation and $\np$ is the market price.\footnote{%
The choice of buying when $\nb\ge \nz$ and selling when $\ns<\nz$ models a taker that is slightly more inclined to sell rather than buy: in this case, the taker is willing to sell even when their valuation is exactly equal to the bid. 
This choice is completely immaterial for the results that follow and is merely done for the sake of simplicity. 
All the results we present still hold if trades happen according to a similar rule but with a taker that is slightly more inclined to buy rather than sell (i.e., if a trade occurs whenever $\nb > \nz$ or $\ns\le\nz$) or if this inclination changes arbitrarily whenever a new taker comes at any new time step.}
Our online market making protocol is specified in \Cref{a:trading-protocol}.
At every round $t$, a taker arrives with a private valuation $\nZ_t \in [0,1]$ and the maker posts bid/ask prices $(\nB_t, \nS_t) \in \cU$.
If $\nB_t < \nZ_t \le \nS_t$ (i.e., if the taker is not willing to sell nor to buy at the proposed bid/ask prices), then no trade happens.
If $\nB_t \ge \nZ_t$ (buy) or $\nS_t < \nZ_t$ (sell), a trade happens.
At the end of each round, the maker observes the market price $\nP_t$ and the type of trade (buy, sell, none) that took place in that round.
For each $(b,a) \in \cU$ and each round $t$, defining $U_t(b,a) \coloneqq u(b,a,\nP_t,\nZ_t)$, the maker's utility at round $t$ is  $\Util_t(\nB_t,\nS_t) = \util(\nB_t,\nS_t,\nP_t,\nZ_t)$. 
Hence, any (possibly randomized) learning strategy for the maker at time $t$ is a map that takes as input past feedback $\big\{ \big(\one{\nB_{\tau} \ge \nZ_{\tau}}, \one{\nS_{\tau} < \nZ_{\tau}}, \nP_{\tau} \big)\big\}_{\tau < t}$, plus some optional random seeds, and returns as output a (random) pair $(\nB_t,\nS_t)$ of bid/ask prices in $\cU$.

The maker's goal is to minimize, for any time horizon $T\in\N$, the regret after $T$ time steps:
\[
    R_T 
\ceq
    \sup_{(\nb, \ns)\in \cU} 
    \E\lsb{ \sum_{t=1}^T \Util_t(\nb, \ns) }
    -
    \E\lsb{ \sum_{t=1}^T \Util_t(\nB_t, \nS_t) } \;,
\]
where the expectations are with respect to the (possible) randomness of the market prices, takers' valuations, and (possibly) the internal randomization of the algorithm.

We study this problem under a mix of several standard assumptions on the underlying process $(\nZ_t, \nP_t)_{t\in\N}$ modeling private valuations and market prices.
\begin{itemize}
    \item[(adv)] $(\nZ_t, \nP_t)_{t\in\N}$ is a stochastic process.
    \item[(iid)] $(\nZ_t, \nP_t)_{t\in\N}$ is an i.i.d.\ stochastic process.\footnote{Note that this assumption alone does not require that, for any $t\in\N$, $\nZ_t$ is independent of $\nP_t$.}
    \item[(iv)] $(\nZ_t, \nP_t)_{t\in\N}$ is a stochastic process such that, for all $t \in \N$, $\nZ_t$ is independent of $\nP_t$.\footnote{Note that this assumption alone does not require that $(\nZ_t, \nP_t)_{t\in\N}$ is an independent stochastic process.}
    \item[(lip)] $(\nZ_t, \nP_t)_{t\in\N}$ is a stochastic process such that, there exists $L > 0$ such that, for all $t \in \N$, the cumulative distribution function of $\nZ_t$ is $L$-Lipschitz.
\end{itemize}

The variables $\nS_t,\nB_t,\nZ_t,\nP_t$ can be thought of as relative to a global reference value revealed at the beginning of each time step, aligning with the convention of quoting in ``points" (or ``ticks") above or below a benchmark. 
With this interpretation, in the regret definition, the benchmark strategy maintains a constant distance from a reference throughout the trading period (as opposed to a constant absolute level).
These observations greatly broaden the scope of application of our setting.

\section{\texorpdfstring{$T^{2/3}$}{T 2/3} Upper Bound (adv+lip) or (iid+iv)}

\label{s:upper}

Since the learner's action space contains a continuum of actions, 
standard bandit approaches require Lipschitzness (or the weaker one-side Lipschitzness) of the reward function, a property that is unfortunately missing in our setting, even when the environment is stable (iid) and the takers are uninformed (iv).
One way to recover Lipschitzness (although only in expectation) is by assuming the takers' valuations are smoooth, i.e., by assuming that their cumulative distribution functions are Lipschitz (lip).
Under the (lip) assumption, recalling that our action space $\cU$ is the set of bid/ask pairs $\{(\nb,\ns) \in [0,1]^2 \mid \nb \le \ns\}$, and noticing that the feedback the learner receives at the end of each time step is sufficient to reconstruct bandit feedback, the market making problem can be seen as an instance of $2$-dimensional Lipschitz bandits.
However, a black-box application of existing $d$-dimensional Lipschitz bandits techniques \citep[with $d=2$]{slivkins2019introduction} would give only a suboptimal regret rate of $\cO(T^{\nicefrac{(d+1)}{(d+2)}}) = \cO(T^{\nicefrac{3}{4}})$.

Using a different approach that exploits the structure of our feedback and reward functions, we manage not only to solve the non-Lipschitz (iid+iv) setting, but also to effectively reduce the dimensionality of the problem by $1$ under the smoothness assumption, obtaining an optimal $\cO(T^{\nicefrac{2}{3}})$ regret rate in all cases.
%
%
To do so, we first note that the utility of our market-making problem can be viewed as the sum of the utilities of two related sub-problems: the first addend in~\eqref{e:utility-of-market-maker} corresponds to the utility in a repeated first-price auction problem with unknown valuations, while the second addend corresponds to the utility in a dynamic pricing problem with unknown costs (see below).
This suggests trying to use two algorithms for the two problems to solve our market-making problem.
However, in our problem we have the further constraint that the bid we propose in the first-price auction is no greater than the price we propose in the dynamic pricing problem.
This constraint prevents us from running directly two independent no-regret algorithms for the two different problems because the bid/ask pair produced this way may violate the constraint.
To get around this roadblock, we present a meta-algorithm (\Cref{a:meta-vin}) that takes as input two algorithms for the two sub-problems and combines them to post pairs in $\cU$ while preserving the respective guarantees.
We now formally introduce the two related problems and build the explicit reduction.

\paragraph{Repeated first-price auctions with unknown valuations (FPA).}
Consider the following online problem (denoted by FPA) of repeated first-price auctions where the learner only learns the valuation of the good they are bidding on \emph{after} the auction is cleared. 
At any time $t$, two $[0,1]$-valued random variables $\nV_t$ and $\nM_t$ are generated and hidden from the learner: $\nV_t$ is the current valuation of the good and $\nM_t$ is the highest competing bid for the good, both unknown to the learner at the time of bidding.
Then, the learner bids an amount $\nX_t \in [0,1]$ and wins the bid if and only if their bid is higher than or equal to the highest competing bid. 
The utility of the learner is $\nV_t - \nX_t$ if they win the auction and zero otherwise.
Finally, the valuation of the good $\nV_t$ and the indicator function $\one{ \nX_t \ge \nM_t }$ (representing whether or not the learner won the auction) are revealed to the learner.

The goal is to minimize, for any time horizon $T\in\N$, the regret after $T$ time steps
\[
    \Rfpa_T
\ceq
    \sup_{\nx \in [0,1] } 
    \E\lsb{ \sum_{t=1}^T (\nV_t - \nx)\,\one{\nx \ge \nM_t } }
    -
    \E\lsb{ \sum_{t=1}^T (\nV_t - \nX_t)\,\one{\nX_t \ge \nM_t } } \;,
\]
where the expectations are taken with respect to the (possible) randomness of $(\nV_1, \nM_1), \dots, (\nV_T, \nM_T)$ and the (possible) internal randomization of the algorithm that outputs the bids $\nX_1, \dots, \nX_T$.

The same setting was studied by~\citet{cesabianchiRoleTransparencyRepeated2024}, but with different feedback models.

\paragraph{Dynamic pricing with unknown costs (DP).}

Consider the following online problem (denoted by DP) of dynamic pricing where the learner only learns the current cost of the item they are selling \emph{after} interacting with the buyer. 
At any time $t$, two $[0,1]$-valued random variables $\nC_t$ and $\nW_t$ are generated and hidden from the learner: $\nC_t$ is the current cost of the item, and $\nW_t$ is the buyer's valuation for the item.
Then, the learner posts a price $\nY_t \in [0,1]$ and the buyer buys the item if and only if their valuation is higher than the posted price. 
The utility of the learner is $\nY_t - \nC_t$ if the buyer bought the item and zero otherwise.
Finally, the cost of the item $\nC_t$ and the indicator function $\one{\nY_t < \nW_t }$ (representing whether or not the buyer bought the item) are revealed to the learner.

The goal is to minimize, for any time horizon $T\in\N$, the regret after $T$ time steps
\[
    \Rdp_T
\ceq
    \sup_{\ny \in [0,1] } 
    \E\lsb{ \sum_{t=1}^T (\ny - \nC_t) \, \one{ \ny < \nW_t } }
    -
    \E\lsb{ \sum_{t=1}^T ( \nY_t - \nC_t) \, \one{ \nY_t < \nW_t } } \;,
\]
where expectations are with respect to the (possible) randomness of $( \nC_1,  \nW_1), \dots, ( \nC_T,  \nW_T)$ and the (possible) internal randomization of the algorithm that outputs the posted prices $ \nY_1, \dots,  \nY_T$.

Dynamic pricing is generally studied under the assumption that costs are known and all equal to zero, and assuming that the buyer buys whenever their evaluation is greater than \emph{or equal to} the posted price---see, e.g.,~\citep{kleinberg2003value,cesa2019dynamic,leme2021learning,luo2024distribution}.

\paragraph{Solving FPA and DP solves market making.}
In this section, we introduce a meta-algorithm that combines two sub-algorithms, one for repeated first-price auctions with unknown valuations and one for dynamic pricing with unknown costs.
We thus obtain an algorithm for market-making and we show in \Cref{t:upper-bound-vin} that its regret can be upper bounded by the sum of the regrets of the two sub-algorithms for the two corresponding sub-problems.
The idea of our meta-algorithm \vin{} (\Cref{a:meta-vin}) begins with the observation that the utility of the market maker is a sum of two terms that correspond to the utilities of the learner in first-price auctions and dynamic pricing, respectively.
This suggests maintaining two algorithms in parallel, one to determine bid prices and one to determine ask prices.
\vin{} then enforces the constraint that bid prices should be no larger than ask prices by swapping the recommendations of the two sub-algorithms whenever they generate corresponding bid/ask prices that would violate the constraint.
Finally, \vin{} leverages the available feedback $\lrb{ \one{\nB_t \ge \nZ_t}, \one{\nS_t < \nZ_t}, \nP_t }$ at time $t$ to reconstruct and relay back to the two sub-algorithms the counterfactual feedback they would have observed if the swap didn't happen (i.e., if the learner always posted the prices determined by the sub-algorithms).
\begin{algorithm2e}[t]
    \Input{Algorithm $\cA^{\mathrm{fpa}}$ for FPA and algorithm $\cA^{\mathrm{dp}}$ for DP}
    \For{%
        time $t=1,2,\ldots$
    }{
        Let $\nX_t$ be the output of $\cA^{\mathrm{fpa}}$ at time $t$ and let $\nY_t$ be the output of $\cA^{\mathrm{dp}}$ at time $t$\;
        \lIf{$\nX_t \le \nY_t$}{%
            let $\nB_t := \nX_t$ and $\nS_t := \nY_t$%
        } \lElseIf{$\nX_t > \nY_t$}{%
            let $\nB_t := \nY_t$ and $\nS_t := \nX_t$%
        }
        Post bid/ask prices $(\nB_t, \nS_t)$ and
        observe feedback $\one{\nB_t \ge \nZ_t}$, $\one{\nS_t < \nZ_t}$, and $\nP_t$\;
        \lIf{$\nX_t \le \nY_t$}{%
            feed $(\one{ \nB_t \ge \nZ_t }, \nP_t)$ to $\cA^{\mathrm{fpa}}$ and $(\one{ \nS_t < \nZ_t }, \nP_t)$ to $\cA^{\mathrm{dp}}$%
        } \lElseIf{$\nX_t > \nY_t$}{%
            feed $(1 - \one{\nS_t < \nZ_t}, \nP_t)$ to $\cA^{\mathrm{fpa}}$ and $(1-\one{\nB_t \ge \nZ_t}, \nP_t)$ to $\cA^{\mathrm{dp}}$%
        }
    }
    \caption{ \vin{} (Meta Market Making) \label{a:meta-vin}}
\end{algorithm2e}

The next result shows that the regret of \vin{} is upper bounded by the sum of the regrets of the two sub-algorithms, which in turn leads to optimal regret guarantees for the market making problem. 
\begin{theorem}\label{t:upper-bound-vin}
    In the online market-making problem, the regret of \vin{} run with parameters $\cA^{\mathrm{fpa}}$ and $\cA^{\mathrm{dp}}$ over the sequence $(\nP_t,\nZ_t)_{t \in [T]}$ of market prices and takers' valuations satisfies, for all $T$,
    $
        R_T
    \le
        \Rfpa_T + \Rdp_T
    $,
    where $\Rfpa_T$ (resp., $\Rdp_T$) is the regret of $\cA^{\mathrm{fpa}}$ (resp., $\cA^{\mathrm{dp}}$) in the FPA (resp., DP) problem over the sequence $(\nV_t,\nM_t)_{t\in[T]} \ceq (\nP_t,\nZ_t)_{t\in[T]}$ (resp., $(\nC_t,\nW_t)_{t\in[T]} \ceq (\nP_t,\nZ_t)_{t\in[T]}$).
    Furthermore, for any horizon $T$, there exist $\cA^{\mathrm{fpa}}$ and $\cA^{\mathrm{dp}}$ such that, for any sequence $(\nP_t,\nZ_t)_{t \in [T]}$ of market prices and takers' valuations, if one of the two following conditions holds:
    \begin{enumerate}
        \item \label{t:assOnePrime}
        For each $t \in [T]$, the cumulative distribution function of $\nZ_t$ is $L$-Lipschitz, for some $L>0$. 
        \item \label{t:assTwoPrime}
        The process $( \nP_t , \nZ_t )_{t \in [T]}$ is i.i.d.\ and, for each $t \in [T]$, $\nP_t$ is independent of $\nZ_t$.
    \end{enumerate}
    Then, the regret of \vin{} run with parameters $\cA^{\mathrm{fpa}}$ and $\cA^{\mathrm{dp}}$ satisfies
    $
        R_T
    \le
        c T^{\nicefrac{2}{3}}
    $,
    with $c \le 2L+100$ (resp., $c \le 102$) if \Cref{t:assOnePrime} (resp., \Cref{t:assTwoPrime}) holds.
\end{theorem}

\begin{proofsketch}
     First, notice that the feedback received by the two algorithms is equal to the feedback they would have retrieved in the corresponding instances regardless of whether or not the prices were swapped. 
     Hence, the regret guarantees of the prices posted by the two algorithm on the original sequence still hold true.
     Notice that, for each $t \in \N$, we have that
     \[
        (\textrm{I}_t) \coloneqq U_t(\nB_t,\nS_t)
    \ge
        (\nP_t - \nX_t) \I\{\nX_t \ge \nZ_t \} + (\nY_t - \nP_t) \I\{\nZ_t > \nY_t \} \eqqcolon (\textrm{II}_t)\;.
     \]
    In fact, if $\nX_t \le \nY_t$ the two quantities are equal because $\nB_t = \nX_t$ and $\nS_t = \nY_t$.
    Instead, if $\nY_t < \nX_t$, then $\nB_t = \nY_t$ and $\nS_t = \nX_t$, and a direct computation shows that if $\nZ_t \le \nY_t$ then
            $
                (\textrm{I}_t) 
                 =
                \nP_t - \nY_t
                >
                \nP_t - \nX_t
                =
                (\textrm{II}_t)
            $;
        if $\nY_t < \nZ_t \le \nX_t$ then
            $
                (\textrm{I}_t) 
                 =
                0
                >
                \nY_t - \nX_t
                =
                (\textrm{II}_t)
            $;
        and if $\nX_t < \nZ_t$ then
            $
                (\textrm{I}_t) 
                =
                \nX_t - \nP_t
                >
                \nY_t - \nP_t
                =
                (\textrm{II}_t)
            $.
    It follows that 
    \begin{align*}
        R_T
    &=
        \sup_{(b,a) \in \cU} \E\lsb{ \sum_{t=1}^T U_t(b,a)} - \E\lsb{\sum_{t=1}^T (\mathrm{I}_t) }
    \le
        \sup_{(b,a) \in \cU} \E\lsb{ \sum_{t=1}^T U_t(b,a)} - \E\lsb{\sum_{t=1}^T (\mathrm{II}_t) }
    \\
    &\le
         \sup_{b \in [0,1]} \E\lsb{ \sum_{t=1}^T (\nP_t - b) \I\{b \ge \nZ_t \} } - \E\lsb{\sum_{t=1}^T (\nP_t - \nX_t) \I\{\nX_t \ge \nZ_t \} }
    \\
    &\qquad+
          \sup_{a \in [0,1]} \E\lsb{ \sum_{t=1}^T (a - \nP_t) \I\{a < \nZ_t \}} - \E\lsb{\sum_{t=1}^T (\nY_t - \nP_t) \I\{\nY_t < \nZ_t\} }
    =
        \Rfpa_T + \Rdp_T\;. 
    \end{align*}
    For the second part of the result, $\cA^{\mathrm{fpa}}$ and $\cA^{\mathrm{dp}}$ can be chosen by designing discretized versions of appropriately defined bandit algorithms. 
    In the Lipschitz case, the proof follows from standard arguments showing that the expected reward of any price can be controlled by one of the closest points in the discretization grid.
    The (iid+iv) case is more subtle and it is based on the fact that, although we cannot accurately approximate the expected reward of all prices, under these assumptions, it is possible to approximate the expected rewards of near-optimal prices, which turns out to be sufficient to obtain the result.
    For all missing details, see \Cref{s:missing-details-VIN}.
\end{proofsketch}

\section{\texorpdfstring{$T^{2/3}$}{T 2/3} Lower Bound (iid+lip+iv)}
\label{s:lowah}

In this section, we present the main result of this work, a lower bound showing that no algorithm can achieve a worst-case regret better than $\Omega(T^{\nicefrac{2}{3}})$ even when the sequence of market prices and takers' valuations is i.i.d.\ with a smooth distribution, and market prices and takers' valuations are independent of each other.

To give the reader some intuition on the nuances and peculiarities of this lower bound, we begin by comparing the trade-offs between regret and the amount of received feedback in hard instances of our setting and those of other (better-known) related settings.
The three main ingredients in the standard lower bound approach for online learning with partial feedback are:
\begin{itemize}
    \item \emph{Hard instances}, constructed as ``$\e$-perturbations'' of a base instance obtained by slightly altering the base distributions of some of the random variables drawn by the environment.
    \item The \emph{amount of information} acquired when playing an action, quantified by the Kullback-Leibler (KL) divergence between the distribution of the feedback obtained when selecting the action in the perturbed instance and the corresponding feedback distribution when selecting the action in the base instance---the higher, the better for the learner.
    \item The \emph{regret} of an action, simply measured by the expected instantaneous regret of the action in the perturbed instance---the smaller, the better for the learner.
\end{itemize}

In hard instances of $K$-armed bandits, an arm is drawn uniformly at random from a $K$-sized set and assigned an expected reward $\e$-higher than the other arms in the set. Selecting an optimal arm costs zero regret and provides $\Theta(\e^2)$ bits of information. Selecting a non-optimal arm, instead, provides no information and costs $\e$ regret \citep{auer2002nonstochastic}.
This implies that any algorithm has only one way to acquire information about an arm: playing that arm. A similar situation occurs in other settings, including those with continuous decision spaces.
For example, to acquire information about a bid in a standard hard instance of first-price auctions (or dynamic pricing), the only reasonable option is to bid that amount, incurring regret $\Theta(\e)$ if the bid is non-optimal and acquiring $\Theta(\e^2)$ bits of information if the arm is optimal \citep{cesabianchiRoleTransparencyRepeated2024}. See \Cref{fig:fp-kl-utility} for a pictorial representation of the expected reward and KL divergence in said hard instances.
\begin{figure}[t]
    \centering
    \includegraphics[page=2, scale=0.8]{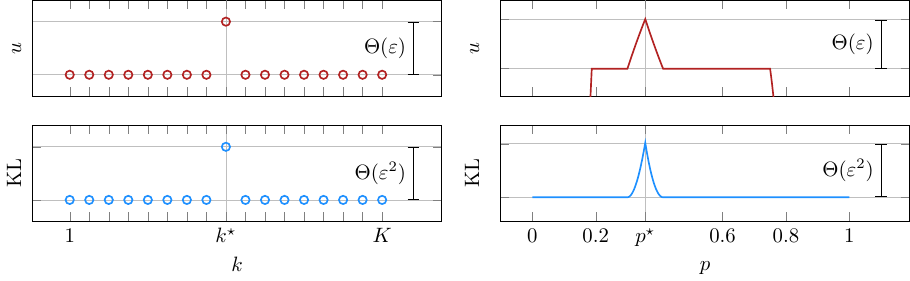}
    \caption{Utility (in red) and amount of acquired information (in blue) in a hard instance of the classic multi-armed bandit problem with a finite set of arms (on the left) and in a hard instance of the first-price auctions problem with a continuous set of bids (on the right)---a similar picture can be drawn for the dynamic pricing problem. 
    Both quantities are maximized at the same point.
    In particular where there is no perturbation, there is no information and some amount of reward, otherwise the KL can grow up to a quantity of order $\e^2$ and the utility can grow by a quantity of order at least $\e$.}
    \label{fig:fp-kl-utility}
\end{figure}

Analogously, to build hard instances of market-making, we begin by drawing a price $\tilde{\nb}$ uniformly at random from a known finite set of bid prices and the bid/ask pair $(\tilde{\nb}, 1)$ is given a reward $\e$-higher than the other (potentially optimal) bid/ask pairs of the form $(\nb',1)$.
However, in the market-making problem, there is a trade-off between reward and feedback that is completely absent in the aforementioned classical problems.
The amount of information and corresponding regret of each pair of prices $(\nb, \ns)$ can be visualized by looking at the qualitative plots in  \Cref{fig:mm-kl-utility}.
\begin{figure}
    \centering
    \includegraphics[width=\linewidth]{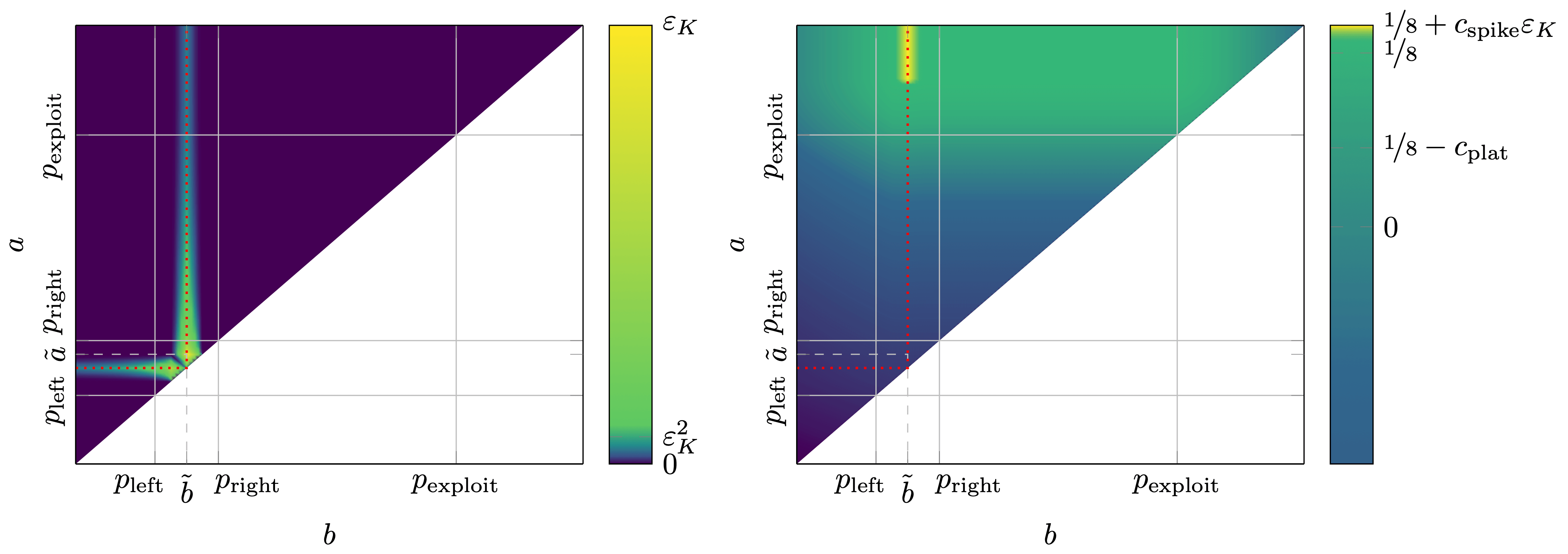}
    \caption{The heat maps of the amount of information (on the left) and the expected utilities (on the right) obtained by selecting pairs of prices in a hard instance we constructed in our $\Omega(T^{\nicefrac{2}{3}})$ lower bound for i.i.d.\ smooth instances with market prices independent of takers' valuations (\Cref{t:lower-bound-lip-iv}). 
    To distinguish the optimal region with reward of (at least) $\nicefrac{1}{8} + c_{\mathrm{spike}}\e_K$ (the yellow region at the top of the right plot) from the suboptimal region with reward $\nicefrac{1}{8}$, one could play any point in the non-blue horizontal and vertical regions on the left plot.
    Note that the pairs that yield the highest amount of feedback (the two small yellow neighborhoods close to the diagonal on the left plot) give highly suboptimal rewards, which are at least $c_{\textrm{plat}} = \Omega(1)$ below $\nicefrac{1}{8}$.
    For explicit definitions of $\e_K, c_{\mathrm{spike}}, c_{\mathrm{plat}}$, see \Cref{s:kl-bounds}.
    }
    \label{fig:mm-kl-utility}
\end{figure}
Here, there are uncountably many ways to determine if a pair of prices $(\nb, 1)$ is optimal: 
playing any pair in the non-blue area of the left plot in \Cref{fig:mm-kl-utility} allows to modulate the amount of exploration, with the points that yield the highest quality feedback being the ones with low reward and high regret (for a plot representing the growth of the amount of feedback and the corresponding reward as a function of the pair $(\nb,\ns)$ taken along a representative path, see \Cref{fig:mm-line-kl-utility}).
Consequently, the exploration-exploitation trade-off consists in \emph{first} choosing whether to exploit with some pair $(\nb,1)$ or instead to explore another potentially optimal pair $(\nb', 1)$, and \emph{then}, in the second case, choosing from an uncountable set of options $(\nb'',\ns'')$ how much regret one is willing to suffer in order to acquire more information about the (potential) optimality of $(\nb',1)$.\footnote{%
The situation is even more complex because pairs $(\nb, \ns)$ allow to test for the optimality not just of the pair $(\nb,1)$, but also of the pair $(\ns,1)$.} 
To build a lower bound, we then need to consider \emph{all} possible ways a learner can decide to leverage these multiple trade-offs between information and regret that arise in the market making problem.

\begin{theorem}\label{t:lower-bound-lip-iv}
    In the online market-making problem, 
    for each $L \ge 8$, each time horizon $T \ge 42$, and each (possibly randomized) algorithm for online market making, there exists a $[0,1]^2$-valued i.i.d.\ sequence $(\nP_t,\nZ_t)_{t \in [T]}$ of market prices and takers' valuations such that for each $t \in [T]$ the two random variables $\nP_t$ and $\nZ_t$ are independent of each other, they admit $L$-Lipschitz cumulative distribution functions, and the regret of the algorithm over the sequence $(\nP_t,\nZ_t)_{t \in [T]}$ is lower bounded by
    $
        R_T \ge c T^{\nicefrac{2}{3}}
    $
    where $c \ge 10^{-6}$ is a constant.
\end{theorem}

We present here a proof sketch and defer the full proof of this result to \Cref{s:kl-bounds}.

\begin{proofsketch}
The first step is to build a base joint distribution over market prices and takers' valuations such that the expected utility function of the learner is maximized over an entire segment of pairs of prices $(\nb,1)$, for $\nb$ that belongs to some interval---note that these are instances where it is never optimal to sell.
Then, we partition this set of maximizers into $K $ 
segments with the same size and build $K$ perturbations of the base distribution such that, in each one of these perturbations, the corresponding expected utilities have a small $\e = \Theta ( 1/K )$-spike inside one of these subsegments.

We then make a simple observation:
If we content ourselves with proving a lower bound for algorithms that play exclusively bid/ask pairs of the form $(\nb, 1)$, our problem reduces to repeated first-price auctions with unknown valuations. 
In this simplified problem, the learner has to essentially locate an $\e$-spike among one of $\nicefrac{1}{\e}$ locations. 
Therefore, the learner can either refuse to locate the spike, and pay an overall $\Theta (\e T)$ regret, or invest $\Omega (\nicefrac{1}{\e^2})$ rounds\footnote{Given that the KL divergence between the base distribution of the feedback of points in a perturbed region and the perturbed distribution of the feedback of the same points is $O(\e^2)$, an information-theoretic argument shows that a sample of size $\Omega(\nicefrac{1}{\e^2})$ in any region is needed to determine whether the spike occurs in that region.} in trying to locate each one of the $\nicefrac{1}{\e}$ potential spikes, paying $\Omega (\e)$ each time a point in the wrong region is sampled, for an overall regret of $\Omega \brb{ \frac{1}{\e^2} \cdot \frac{1}{\e} \cdot \e } = \Omega ( \nicefrac{1}{\e^2} )$. 
The hardest instance is when $\e T = \Theta ( \nicefrac{1}{\e^2} )$, i.e., when $\e = \Theta ( T^{\nicefrac{1}{3}} )$.
In this case, no matter what the learner decides, they will always pay at least $\Omega ( \e T ) = \Omega( \nicefrac{1}{\e^2} ) = \Omega (T^{\nicefrac{2}{3}} )$ regret.

However, the problem is substantially more delicate in the general case where, to explore a potentially optimal action $(\nb, 1)$, the learner has access to uncountably many options $(\nb', \ns')$ (see \Cref{fig:mm-kl-utility}).
This also leads to another complication.
In standard online learning lower bounds, the worst-case regret of an algorithm is analyzed by counting how many times the algorithm plays exploring vs.\ exploiting arms, then quantifying how much information was gathered from the exploration, and finally summing over all exploring arms.
This strategy is hardly implementable in our setting because each exploiting pair $(b, 1)$ can be explored with uncountably many other arms $(\nb', \ns')$, where each arm trades off reward with feedback in a different way (see \Cref{fig:mm-line-kl-utility})---for example, one could post two prices $(\nb, \nb + \theta)$, for some small $\theta>0$ paying $\Omega(1)$ regret to obtain much higher quality feedback than if they played $(\nb, 1)$ ($\e$ vs.\ $\e^2$ bits of information).
\begin{figure}[t]
    \centering
    \includegraphics[page=1, scale=0.8]{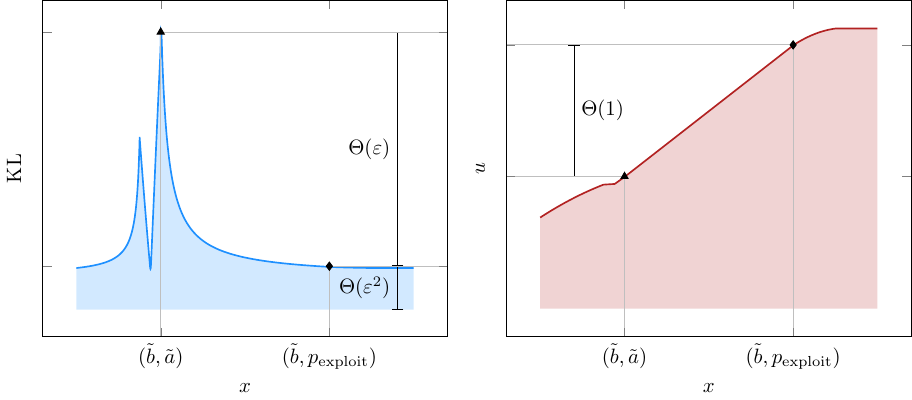}
    \caption{
    The amount of information (on the left) and the corresponding expected utility (on the right) as a function of pairs of prices $(\nb,\ns)$, with $(\nb,\ns)$ following the red dotted line in \Cref{fig:mm-kl-utility} starting at the left boundary of the horizontal non-blue region at the point $(0,\tilde{\nb})$, moving horizontally up to the diagonal, then vertically up to the optimal pair $(\tilde{\nb},1)$, parameterized by a real number $x$.
    The points $(\tilde \nb, \tilde \ns)$ and $(\tilde \nb, p_\mathrm{exploit})$ illustrate how regions of high reward and regions of high information do not coincide, in fact the cost of gathering more information (e.g., play $(\tilde \nb, \tilde a)$ instead of $(\tilde \nb, p_\mathrm{exploit})$) can be a constant $\Theta(1)$ in the reward. On the other hand, a high-reward play gathers only $\Theta(\varepsilon^2)$ amount of information, hence costing $\Theta(\varepsilon)$ with respect to the maximum amount of information that can be gathered.
    }
    \label{fig:mm-line-kl-utility}
\end{figure}

We circumvent this problem by clustering the action set $\cU$ into a growing (with the time horizon $T$) number of disjoint subsets of non-uniform sizes that we further group in three macro-categories: ``suffer high regret but (potentially) obtain a large amount of information'', ``suffer little (or potentially no) regret but also obtain little information'', and clearly suboptimal actions where we ``suffer some regret and obtain no information''.
We then analyze the regret of any algorithm in terms of how many times it selects an action from each cluster, using similar techniques for all clusters belonging to the same macro-category.
By doing so, we can prove that no matter how an algorithm decides to play, there will always be instances where it has to pay a regret of order at least $\Omega (T^{\nicefrac{2}{3}})$.
\end{proofsketch}

We conclude this section by remarking that, to the best of our knowledge, this is the first lower bound for a problem exhibiting such a continuum of trade-offs between reward and feedback.
As our setting is arguably simple, we expect this phenomenon to occur in other applications of online learning to digital markets. 
Therefore, we believe the lower bound techniques developed in this work could be of independent interest.

\section{Unlearnability without (lip) or (iv) Assumptions}

In this section, we show that, even under the (iid) assumption, without smoothness (lip) or independence between market prices and takers' valuations (iv), learning is impossible in general.

\begin{theorem}\label{t:minimal-lower-bound-iid}
In the online market-making problem, 
for each time horizon $T$ and each $\delta \in (0, \nicefrac{1}{2})$ and each (possibly randomized) algorithm for online market making, there exists a $[0,1]^2$-valued i.i.d.\ sequence $(\nP_t,\nZ_t)_{t \in [T]}$ of market prices and takers' valuations such that the regret of the algorithm over the sequence $(\nP_t,\nZ_t)_{t \in [T]}$ is lower bounded by
$
    R_T \ge \lrb{\frac{1}{2} - \delta} \cdot T
$.
\end{theorem}

We present here a proof sketch and defer the full proof of this result to \Cref{s:liner-lower-bound-appe}.

\begin{proofsketch}
In an i.i.d.\ setting, without any loss of generality, we can assume that the learning algorithm is deterministic.
Notice that, if market prices are always equal to $0$ or $1$, a deterministic algorithm can only generate finitely many points over a finite time horizon.
Therefore, for any fixed deterministic algorithm and any small $0< \e \approx 0$, there exists a small interval $[c,d]$ included in $\bsb{ \frac12 - \e, \frac12 + \e }$ where the algorithm never plays over a finite time horizon.
By building an i.i.d.\ family of market prices and takers' valuations such that $(\nP_t,\nZ_t) = (0,d)$ or $(\nP_t,\nZ_t) = (1,c)$ with probability $\nicefrac{1}{2}$ each, one can prove that (i) the best fixed bid/ask pair belongs to the open interval $(c,d)$, (ii) that this pair always wins (roughly) at least $1/2$ reward in expectation, while (iii) the learning algorithm (roughly) gains $0$ in expectation.   
\end{proofsketch}

\section{Conclusions and Future Directions}
We initiated an investigation of market making under an online learning framework and provided tight bounds on the regret under various natural assumptions. 
While the regret of various problems related to financial trading was previously investigated, most of these results are from the viewpoint of a liquidity seeker. 
Liquidity providers are crucial for the functioning of markets, giving practical impact to our theoretical findings.

Extending our analysis to full feedback, in which the learner observes the private valuations at the end of each round, would shed light on the cost of the more realistic feedback studied in this paper.
In a different direction, we could consider market makers that could clear their position at a later time, which would make our setting more closely related to \textit{bandits-with-knapsacks} \citep{badanidiyuru_bandits_2018}, which is known to be unsolvable in the adversarial setting \citep{immorlica_adversarial_2023}.

\acks{
    NCB is partially supported by the EU Horizon CL4-2021-HUMAN-01 research and innovation action under grant agreement 101070617, project ELSA (European Lighthouse on Secure and Safe AI), by the MUR PRIN grant 2022EKNE5K (Learning in Markets and Society), and by the One Health Action Hub, University Task Force for the resilience of territorial ecosystems, funded by Università degli Studi di Milano (PSR 2021-GSA-Linea 6).
    
    RC is partially supported by the FAIR (Future Artificial Intelligence Research) project, funded by the NextGenerationEU program within the PNRR{-}PE{-}AI scheme  (M4C2, investment 1.3, line on Artificial Intelligence).
    
    LF is partially supported by the EU Horizon CL4{-}2022{-}HUMAN{-}02 research and innovation action under grant agreement 101120237, project ELIAS (European Lighthouse of AI for Sustainability).

    TC gratefully acknowledges the support of the University of Ottawa through grant GR002837 (Start-Up Funds) and that of the Natural Sciences and Engineering Research Council of Canada (NSERC) through grants RGPIN{-}2023{-}03688 (Discovery Grants Program) and DGECR2023{-}00208 (Discovery Grants Program, DGECR $-$ Discovery Launch Supplement).
}

\clearpage

\bibliography{mybibliography}

\begin{thebibliography}{76}
\providecommand{\natexlab}[1]{#1}
\providecommand{\url}[1]{\texttt{#1}}
\expandafter\ifx\csname urlstyle\endcsname\relax
  \providecommand{\doi}[1]{doi: #1}\else
  \providecommand{\doi}{doi: \begingroup \urlstyle{rm}\Url}\fi

\bibitem[Abernethy and Kale(2013)]{abernethy2013adaptive}
Jacob Abernethy and Satyen Kale.
\newblock Adaptive market making via online learning.
\newblock \emph{Advances in Neural Information Processing Systems}, 26, 2013.

\bibitem[Abernethy et~al.(2012)Abernethy, Frongillo, and Wibisono]{abernethy2012minimax}
Jacob Abernethy, Rafael~M Frongillo, and Andre Wibisono.
\newblock Minimax option pricing meets {B}lack-{S}choles in the limit.
\newblock In \emph{Proceedings of the Forty-Fourth Annual ACM Symposium on Theory of Computing}, pages 1029--1040, 2012.

\bibitem[Abernethy et~al.(2013)Abernethy, Bartlett, Frongillo, and Wibisono]{abernethy2013hedge}
Jacob Abernethy, Peter~L Bartlett, Rafael Frongillo, and Andre Wibisono.
\newblock How to hedge an option against an adversary: {B}lack-{S}choles pricing is minimax optimal.
\newblock \emph{Advances in Neural Information Processing Systems}, 26, 2013.

\bibitem[Achddou et~al.(2021)Achddou, Capp{\'e}, and Garivier]{pmlr-v157-achddou21a}
Juliette Achddou, Olivier Capp{\'e}, and Aur{\'e}lien Garivier.
\newblock Fast rate learning in stochastic first price bidding.
\newblock In Vineeth~N. Balasubramanian and Ivor Tsang, editors, \emph{Proceedings of the 13th asian conference on machine learning}, volume 157 of \emph{Proceedings of Machine Learning Research}, 2021.

\bibitem[Almgren and Chriss(2001)]{almgren2001optimal}
Robert Almgren and Neil Chriss.
\newblock Optimal execution of portfolio transactions.
\newblock \emph{Journal of Risk}, 3:\penalty0 5--40, 2001.

\bibitem[Audibert and Bubeck(2010)]{audibert10a}
Jean-Yves Audibert and S{{\'e}}bastien Bubeck.
\newblock Regret bounds and minimax policies under partial monitoring.
\newblock \emph{Journal of Machine Learning Research}, 11\penalty0 (94):\penalty0 2785--2836, 2010.

\bibitem[Auer et~al.(2002{\natexlab{a}})Auer, Cesa-Bianchi, and Fischer]{auerFinitetimeAnalysisMultiarmed2002}
Peter Auer, Nicol\`{o} Cesa-Bianchi, and Paul Fischer.
\newblock Finite-time analysis of the multiarmed bandit problem.
\newblock \emph{Mach. Learn.}, 2002{\natexlab{a}}.
\newblock \doi{10.1023/A:1013689704352}.

\bibitem[Auer et~al.(2002{\natexlab{b}})Auer, Cesa-Bianchi, Freund, and Schapire]{auer2002nonstochastic}
Peter Auer, Nicolo Cesa-Bianchi, Yoav Freund, and Robert~E Schapire.
\newblock The nonstochastic multiarmed bandit problem.
\newblock \emph{SIAM journal on computing}, 32\penalty0 (1):\penalty0 48--77, 2002{\natexlab{b}}.

\bibitem[Azar et~al.(2022)Azar, Fiat, and Fusco]{azar2022alpha}
Yossi Azar, Amos Fiat, and Federico Fusco.
\newblock An $\alpha $-regret analysis of adversarial bilateral trade.
\newblock \emph{Advances in Neural Information Processing Systems}, 35:\penalty0 1685--1697, 2022.

\bibitem[Babaioff et~al.(2015)Babaioff, Dughmi, Kleinberg, and Slivkins]{babaioffDynamicPricingLimited2015}
Moshe Babaioff, Shaddin Dughmi, Robert Kleinberg, and Aleksandrs Slivkins.
\newblock Dynamic {{Pricing}} with {{Limited Supply}}.
\newblock \emph{ACM Trans. Econ. Comput.}, 3, 2015.
\newblock ISSN 2167-8375.
\newblock \doi{10.1145/2559152}.

\bibitem[Bachoc et~al.(2024)Bachoc, Cesa-Bianchi, Cesari, and Colomboni]{fairgainfromtrade}
Fran\c{c}ois Bachoc, Nicol\`{o} Cesa-Bianchi, Tommaso Cesari, and Roberto Colomboni.
\newblock Fair online bilateral trade.
\newblock In \emph{Advances in Neural Information Processing Systems}, volume~37, pages 37241--37263, 2024.

\bibitem[Bachoc et~al.(2025{\natexlab{a}})Bachoc, Cesari, and Colomboni]{bachoc2025contextual}
François Bachoc, Tommaso Cesari, and Roberto Colomboni.
\newblock A tight regret analysis of non-parametric repeated contextual brokerage.
\newblock In \emph{Proceedings of The 28th International Conference on Artificial Intelligence and Statistics}, Proceedings of Machine Learning Research. PMLR, 2025{\natexlab{a}}.

\bibitem[Bachoc et~al.(2025{\natexlab{b}})Bachoc, Cesari, and Colomboni]{bachoc2025contextualLinear}
François Bachoc, Tommaso Cesari, and Roberto Colomboni.
\newblock A parametric contextual online learning theory of brokerage.
\newblock In \emph{Forty-Second International Conference on Machine Learning}, 2025{\natexlab{b}}.

\bibitem[Badanidiyuru et~al.(2018)Badanidiyuru, Kleinberg, and Slivkins]{badanidiyuru_bandits_2018}
Ashwinkumar Badanidiyuru, Robert Kleinberg, and Aleksandrs Slivkins.
\newblock Bandits with {Knapsacks}.
\newblock \emph{Journal of the ACM}, 2018.
\newblock \doi{10.1145/3164539}.

\bibitem[Bernasconi et~al.(2024)Bernasconi, Castiglioni, Celli, and Fusco]{bernasconi2024no}
Martino Bernasconi, Matteo Castiglioni, Andrea Celli, and Federico Fusco.
\newblock No-regret learning in bilateral trade via global budget balance.
\newblock In \emph{Proceedings of the 56th Annual ACM Symposium on Theory of Computing}, pages 247--258, 2024.

\bibitem[Bertsimas and Lo(1998)]{bertsimas1998optimal}
Dimitris Bertsimas and Andrew~W Lo.
\newblock Optimal control of execution costs.
\newblock \emph{Journal of Financial Markets}, 1\penalty0 (1):\penalty0 1--50, 1998.

\bibitem[Black and Scholes(1973)]{black1973pricing}
Fischer Black and Myron Scholes.
\newblock The pricing of options and corporate liabilities.
\newblock \emph{Journal of Political Economy}, 81\penalty0 (3):\penalty0 637--654, 1973.

\bibitem[Block et~al.(2022)Block, Dagan, Golowich, and Rakhlin]{block2022smoothed}
Adam Block, Yuval Dagan, Noah Golowich, and Alexander Rakhlin.
\newblock Smoothed online learning is as easy as statistical learning.
\newblock In \emph{{COLT}}, volume 178 of \emph{Proceedings of Machine Learning Research}, pages 1716--1786. {PMLR}, 2022.

\bibitem[Blum and Kalai(1997)]{blum1997universal}
Avrim Blum and Adam Kalai.
\newblock Universal portfolios with and without transaction costs.
\newblock In \emph{Proceedings of the Tenth Annual Conference on Computational Learning Theory}, pages 309--313, 1997.

\bibitem[Boli{\'c} et~al.(2024)Boli{\'c}, Cesari, and Colomboni]{bolic2024online}
Natasa Boli{\'c}, Tommaso Cesari, and Roberto Colomboni.
\newblock An online learning theory of brokerage.
\newblock In \emph{Proceedings of the 23rd International Conference on Autonomous Agents and Multiagent Systems}, pages 216--224, 2024.

\bibitem[Borodin et~al.(2000)Borodin, El-Yaniv, and Gogan]{borodin2000competitive}
Allan Borodin, Ran El-Yaniv, and Vincent Gogan.
\newblock On the competitive theory and practice of portfolio selection.
\newblock In \emph{LATIN 2000: Theoretical Informatics: 4th Latin American Symposium, Punta del Este, Uruguay, April 10-14, 2000 Proceedings 4}, pages 173--196. Springer, 2000.

\bibitem[Borodin et~al.(2003)Borodin, El-Yaniv, and Gogan]{borodin2003can}
Allan Borodin, Ran El-Yaniv, and Vincent Gogan.
\newblock Can we learn to beat the best stock.
\newblock \emph{Advances in Neural Information Processing Systems}, 16, 2003.

\bibitem[Cartea et~al.(2015)Cartea, Jaimungal, and Penalva]{cartea2015algorithmic}
{\'A}lvaro Cartea, Sebastian Jaimungal, and Jos{\'e} Penalva.
\newblock \emph{Algorithmic and high-frequency trading}.
\newblock Cambridge University Press, 2015.

\bibitem[Cesa-Bianchi and Lugosi(2006)]{10.5555/1137817}
Nicolo Cesa-Bianchi and Gabor Lugosi.
\newblock \emph{Prediction, Learning, and Games}.
\newblock Cambridge University Press, USA, 2006.
\newblock ISBN 0521841089.

\bibitem[Cesa-Bianchi et~al.(2013)Cesa-Bianchi, Gentile, and Mansour]{cesaSecondPrice}
Nicol\`{o} Cesa-Bianchi, Claudio Gentile, and Yishay Mansour.
\newblock Regret minimization for reserve prices in second-price auctions.
\newblock In \emph{Proceedings of the Twenty-Fourth Annual ACM-SIAM Symposium on Discrete Algorithms}, SODA '13, page 1190–1204, USA, 2013. Society for Industrial and Applied Mathematics.
\newblock ISBN 9781611972511.

\bibitem[Cesa-Bianchi et~al.(2019)Cesa-Bianchi, Cesari, and Perchet]{cesa2019dynamic}
Nicolo Cesa-Bianchi, Tommaso Cesari, and Vianney Perchet.
\newblock Dynamic pricing with finitely many unknown valuations.
\newblock In \emph{Algorithmic Learning Theory}, pages 247--273. PMLR, 2019.

\bibitem[Cesa-Bianchi et~al.(2021)Cesa-Bianchi, Cesari, Colomboni, Fusco, and Leonardi]{cesa2021regret}
Nicol{\`o} Cesa-Bianchi, Tommaso~R Cesari, Roberto Colomboni, Federico Fusco, and Stefano Leonardi.
\newblock A regret analysis of bilateral trade.
\newblock In \emph{Proceedings of the 22nd ACM Conference on Economics and Computation}, pages 289--309, 2021.

\bibitem[Cesa-Bianchi et~al.(2023)Cesa-Bianchi, Cesari, Colomboni, Fusco, and Leonardi]{cesa2023repeated}
Nicol{\`o} Cesa-Bianchi, Tommaso~R Cesari, Roberto Colomboni, Federico Fusco, and Stefano Leonardi.
\newblock Repeated bilateral trade against a smoothed adversary.
\newblock In \emph{The Thirty Sixth Annual Conference on Learning Theory}, pages 1095--1130. PMLR, 2023.

\bibitem[Cesa-Bianchi et~al.(2024{\natexlab{a}})Cesa-Bianchi, Cesari, Colomboni, Fusco, and Leonardi]{JMLR:v25:23-1627}
Nicol{{\`o}} Cesa-Bianchi, Tommaso Cesari, Roberto Colomboni, Federico Fusco, and Stefano Leonardi.
\newblock Regret analysis of bilateral trade with a smoothed adversary.
\newblock \emph{Journal of Machine Learning Research}, 25\penalty0 (234):\penalty0 1--36, 2024{\natexlab{a}}.

\bibitem[Cesa-Bianchi et~al.(2024{\natexlab{b}})Cesa-Bianchi, Cesari, Colomboni, Fusco, and Leonardi]{cesa2024bilateral}
Nicol{\`o} Cesa-Bianchi, Tommaso Cesari, Roberto Colomboni, Federico Fusco, and Stefano Leonardi.
\newblock Bilateral trade: A regret minimization perspective.
\newblock \emph{Mathematics of Operations Research}, 49\penalty0 (1):\penalty0 171--203, 2024{\natexlab{b}}.

\bibitem[Cesa-Bianchi et~al.(2024{\natexlab{c}})Cesa-Bianchi, Cesari, Colomboni, Fusco, and Leonardi]{cesabianchiRoleTransparencyRepeated2024}
Nicol{\`o} Cesa-Bianchi, Tommaso Cesari, Roberto Colomboni, Federico Fusco, and Stefano Leonardi.
\newblock The role of transparency in repeated first-price auctions with unknown valuations.
\newblock In \emph{Proceedings of the 56th Annual ACM Symposium on Theory of Computing}, pages 225--236, 2024{\natexlab{c}}.

\bibitem[Cesari and Colomboni(2025)]{cesari2025volumeFairnesss}
Tommaso Cesari and Roberto Colomboni.
\newblock An online learning theory of trading-volume maximization.
\newblock In \emph{The Thirteenth International Conference on Learning Representations}, 2025.

\bibitem[Cesa‐Bianchi et~al.(2025)Cesa‐Bianchi, Colomboni, and Kasy]{taxation2025}
Nicolò Cesa‐Bianchi, Roberto Colomboni, and Maximilian Kasy.
\newblock Adaptive maximization of social welfare.
\newblock \emph{Econometrica}, 93\penalty0 (3):\penalty0 1073--1104, 2025.
\newblock \doi{https://doi.org/10.3982/ECTA22351}.

\bibitem[Cover and Thomas(2006)]{ThomasM.Cover2006}
T.~M. Cover and J.~A. Thomas.
\newblock \emph{Elements of Information Theory}.
\newblock Wiley-Interscience, 2006.

\bibitem[Cover(1991)]{coverUniversalPortfolios1991}
Thomas~M. Cover.
\newblock Universal portfolios.
\newblock \emph{Mathematical Finance}, 1991.

\bibitem[Cover and Ordentlich(1996)]{cover1996universal}
Thomas~M Cover and Erik Ordentlich.
\newblock Universal portfolios with side information.
\newblock \emph{IEEE Transactions on Information Theory}, 42\penalty0 (2):\penalty0 348--363, 1996.

\bibitem[Das(2005)]{learnigmarketmakerGlostenMilgrom2005}
Sanmay Das.
\newblock A learning market-maker in the glosten–milgrom model.
\newblock \emph{Quantitative Finance}, 5:\penalty0 169--180, 04 2005.
\newblock \doi{10.1080/14697680500148067}.

\bibitem[DeMarzo et~al.(2006)DeMarzo, Kremer, and Mansour]{demarzo2006online}
Peter DeMarzo, Ilan Kremer, and Yishay Mansour.
\newblock Online trading algorithms and robust option pricing.
\newblock In \emph{Proceedings of the Thirty-Eighth Annual ACM Symposium on Theory of computing}, pages 477--486, 2006.

\bibitem[Durvasula et~al.(2023)Durvasula, Haghtalab, and Zampetakis]{DurvasulaHZ23}
Naveen Durvasula, Nika Haghtalab, and Manolis Zampetakis.
\newblock Smoothed analysis of online non-parametric auctions.
\newblock In \emph{{EC}}, pages 540--560. {ACM}, 2023.

\bibitem[Fender and Lewrick(2015)]{fender2015early}
Ingo Fender and Ulf Lewrick.
\newblock Shifting tides --– market liquidity and market-making in fixed income instruments.
\newblock \emph{BIS Quarterly Review, March}, 2015.

\bibitem[Flajolet and Jaillet(2017)]{NIPS2017_0bed45bd}
Arthur Flajolet and Patrick Jaillet.
\newblock Real-time bidding with side information.
\newblock In I.~Guyon, U.~Von Luxburg, S.~Bengio, H.~Wallach, R.~Fergus, S.~Vishwanathan, and R.~Garnett, editors, \emph{Advances in neural information processing systems}, volume~30, 2017.
\newblock URL \url{https://proceedings.neurips.cc/paper_files/paper/2017/file/0bed45bd5774ffddc95ffe500024f628-Paper.pdf}.

\bibitem[Gatheral(2011)]{gatheral2011volatility}
J~Gatheral.
\newblock \emph{The volatility surface: A Practitioner's Guide}.
\newblock John Wiley and Sons, Inc, 2011.

\bibitem[Gaucher et~al.(2025)Gaucher, Bernasconi, Castiglioni, Celli, and Perchet]{GaucherBCCP25}
Solenne Gaucher, Martino Bernasconi, Matteo Castiglioni, Andrea Celli, and Vianney Perchet.
\newblock Feature-based online bilateral trade.
\newblock In \emph{The Thirteenth International Conference on Learning Representations}, 2025.

\bibitem[Glosten and Milgrom(1985)]{glostenBidAskTransaction1985}
Lawrence~R. Glosten and Paul~R. Milgrom.
\newblock Bid, ask and transaction prices in a specialist market with heterogeneously informed traders.
\newblock \emph{Journal of Financial Economics}, 14, 1985.
\newblock ISSN 0304405X.
\newblock \doi{10.1016/0304-405X(85)90044-3}.

\bibitem[Haghtalab et~al.(2020)Haghtalab, Roughgarden, and Shetty]{haghtalab2020smoothed}
Nika Haghtalab, Tim Roughgarden, and Abhishek Shetty.
\newblock Smoothed analysis of online and differentially private learning.
\newblock In \emph{NeurIPS}, 2020.

\bibitem[Haghtalab et~al.(2021)Haghtalab, Roughgarden, and Shetty]{HaghtalabRS21}
Nika Haghtalab, Tim Roughgarden, and Abhishek Shetty.
\newblock Smoothed analysis with adaptive adversaries.
\newblock In \emph{{FOCS}}, pages 942--953. {IEEE}, 2021.

\bibitem[Haghtalab et~al.(2022)Haghtalab, Han, Shetty, and Yang]{haghtalaboracle}
Nika Haghtalab, Yanjun Han, Abhishek Shetty, and Kunhe Yang.
\newblock Oracle-efficient online learning for smoothed adversaries.
\newblock In \emph{NeurIPS}, 2022.

\bibitem[Hakansson(1975)]{hakansson1975optimal}
Nils~H Hakansson.
\newblock Optimal investment and consumption strategies under risk for a class of utility functions.
\newblock In \emph{Stochastic Optimization Models in Finance}, pages 525--545. Elsevier, 1975.

\bibitem[Han et~al.(2024)Han, Zhou, and Weissman]{han2024optimalnoregretlearningrepeated}
Yanjun Han, Zhengyuan Zhou, and Tsachy Weissman.
\newblock Optimal no-regret learning in repeated first-price auctions, 2024.

\bibitem[Harris(2003)]{harris2003trading}
Larry Harris.
\newblock \emph{Trading and exchanges: Market microstructure for practitioners}.
\newblock OUP USA, 2003.

\bibitem[Hazan(2016)]{hazan2016introduction}
Elad Hazan.
\newblock Introduction to online convex optimization.
\newblock \emph{Foundations and Trends{\textregistered} in Optimization}, 2\penalty0 (3-4):\penalty0 157--325, 2016.

\bibitem[Helmbold et~al.(1998)Helmbold, Schapire, Singer, and Warmuth]{helmbold1998line}
David~P Helmbold, Robert~E Schapire, Yoram Singer, and Manfred~K Warmuth.
\newblock On-line portfolio selection using multiplicative updates.
\newblock \emph{Mathematical Finance}, 8\penalty0 (4):\penalty0 325--347, 1998.

\bibitem[Hull(2017)]{hull2017options}
J.C. Hull.
\newblock \emph{Options, Futures, and Other Derivatives}.
\newblock Pearson Education, 2017.
\newblock ISBN 9780134631493.

\bibitem[Immorlica et~al.(2023)Immorlica, Sankararaman, Schapire, and Slivkins]{immorlica_adversarial_2023}
Nicole Immorlica, Karthik~Abinav Sankararaman, Robert Schapire, and Aleksandrs Slivkins.
\newblock Adversarial {Bandits} with {Knapsacks}, 2023.

\bibitem[Kannan et~al.(2018)Kannan, Morgenstern, Roth, Waggoner, and Wu]{kannan2018smoothed}
Sampath Kannan, Jamie~H Morgenstern, Aaron Roth, Bo~Waggoner, and Zhiwei~Steven Wu.
\newblock A smoothed analysis of the greedy algorithm for the linear contextual bandit problem.
\newblock \emph{Advances in neural information processing systems}, 31, 2018.

\bibitem[Kleinberg and Leighton(2003)]{kleinberg2003value}
Robert Kleinberg and Tom Leighton.
\newblock The value of knowing a demand curve: Bounds on regret for online posted-price auctions.
\newblock In \emph{44th Annual IEEE Symposium on Foundations of Computer Science, 2003. Proceedings.}, pages 594--605. IEEE, 2003.

\bibitem[Kleinberg(2004)]{Kleinberg2004NearlyTB}
Robert~D. Kleinberg.
\newblock Nearly tight bounds for the continuum-armed bandit problem.
\newblock In \emph{Neural Information Processing Systems}, 2004.

\bibitem[Konishi(2002)]{konishi2002optimal}
Hizuru Konishi.
\newblock Optimal slice of a vwap trade.
\newblock \emph{Journal of Financial Markets}, 5\penalty0 (2):\penalty0 197--221, 2002.

\bibitem[Lattimore and Szepesv{\'a}ri(2020)]{lattimore2020bandit}
T.~Lattimore and C.~Szepesv{\'a}ri.
\newblock \emph{Bandit Algorithms}.
\newblock Cambridge University Press, 2020.
\newblock ISBN 9781108486828.

\bibitem[Leme et~al.(2021)Leme, Sivan, Teng, and Worah]{leme2021learning}
Renato~Paes Leme, Balasubramanian Sivan, Yifeng Teng, and Pratik Worah.
\newblock Learning to price against a moving target.
\newblock In \emph{International Conference on Machine Learning}, pages 6223--6232. PMLR, 2021.

\bibitem[Luo et~al.(2024)Luo, Sun, and Liu]{luo2024distribution}
Yiyun Luo, Will~Wei Sun, and Yufeng Liu.
\newblock Distribution-free contextual dynamic pricing.
\newblock \emph{Mathematics of Operations Research}, 49\penalty0 (1):\penalty0 599--618, 2024.

\bibitem[MacLean et~al.(2012)MacLean, Thorp, and Ziemba]{maclean2011kelly}
L.C. MacLean, E.O. Thorp, and W.T. Ziemba.
\newblock \emph{The {K}elly Capital Growth Investment Criterion: Theory and Practice}.
\newblock Handbook in Financial Economics. World Scientific, 2012.
\newblock ISBN 9789814293495.

\bibitem[Markowitz(1952)]{markowitz}
Harry Markowitz.
\newblock Portfolio selection.
\newblock \emph{The Journal of Finance}, 7\penalty0 (1):\penalty0 77--91, 1952.
\newblock ISSN 00221082, 15406261.

\bibitem[Orabona(2019)]{orabona2019modern}
Francesco Orabona.
\newblock A modern introduction to online learning.
\newblock \emph{arXiv preprint arXiv:1912.13213}, 2019.

\bibitem[Rakhlin et~al.(2011)Rakhlin, Sridharan, and Tewari]{rakhlin2011online}
Alexander Rakhlin, Karthik Sridharan, and Ambuj Tewari.
\newblock Online learning: Stochastic, constrained, and smoothed adversaries.
\newblock In \emph{{NIPS}}, 2011.

\bibitem[Shalev-Shwartz(2012)]{shalev2012online}
Shai Shalev-Shwartz.
\newblock Online learning and online convex optimization.
\newblock \emph{Foundations and Trends{\textregistered} in Machine Learning}, 4\penalty0 (2):\penalty0 107--194, 2012.

\bibitem[Shalunov et~al.(2020)Shalunov, Kitaev, Shalunov, and Akopyan]{shalunov2020calculated}
Stanislav Shalunov, Alexei Kitaev, Yakov Shalunov, and Arseniy Akopyan.
\newblock Calculated boldness.
\newblock \emph{arXiv preprint arXiv:2012.13830}, 2020.

\bibitem[Shreve(2004)]{shreve2004stochastic}
Steven Shreve.
\newblock \emph{Stochastic calculus for finance II: Continuous-time models}, volume~11.
\newblock Springer, 2004.

\bibitem[Shreve(2005)]{shreve2005stochastic}
Steven Shreve.
\newblock \emph{Stochastic calculus for finance I: the binomial asset pricing model}.
\newblock Springer Science \& Business Media, 2005.

\bibitem[Slivkins(2019)]{slivkins2019introduction}
Aleksandrs Slivkins.
\newblock Introduction to multi-armed bandits.
\newblock \emph{Foundations and Trends{\textregistered} in Machine Learning}, 12\penalty0 (1-2):\penalty0 1--286, 2019.

\bibitem[Spielman and Teng(2004)]{spielman2004smoothed}
Daniel~A Spielman and Shang-Hua Teng.
\newblock Smoothed analysis of algorithms: Why the simplex algorithm usually takes polynomial time.
\newblock \emph{Journal of the ACM (JACM)}, 51\penalty0 (3):\penalty0 385--463, 2004.
\newblock \doi{10.1145/990308.990310}.

\bibitem[Thorp(1969)]{breiman1961optimal}
E.~O. Thorp.
\newblock Optimal gambling systems for favorable games.
\newblock \emph{Revue de l'Institut International de Statistique / Review of the International Statistical Institute}, 1969.
\newblock ISSN 03731138.

\bibitem[Thorp(1975)]{thorp1975portfolio}
Edward~O Thorp.
\newblock Portfolio choice and the kelly criterion.
\newblock In \emph{Stochastic optimization models in finance}, pages 599--619. Elsevier, 1975.

\bibitem[Tsybakov(2008)]{Tsybakov2008}
Alexandre~B. Tsybakov.
\newblock \emph{Introduction to Nonparametric Estimation}.
\newblock Springer Publishing Company, Incorporated, 2008.

\bibitem[Weed et~al.(2015)Weed, Perchet, and Rigollet]{weedOnlineLearningRepeated2015}
Jonathan Weed, Vianney Perchet, and Philippe Rigollet.
\newblock Online learning in repeated auctions, 2015.

\bibitem[Zimmert et~al.(2022)Zimmert, Agarwal, and Kale]{zimmert2022pushing}
Julian Zimmert, Naman Agarwal, and Satyen Kale.
\newblock Pushing the efficiency-regret pareto frontier for online learning of portfolios and quantum states.
\newblock In \emph{Conference on Learning Theory}, pages 182--226. PMLR, 2022.

\end{thebibliography}

\clearpage

\appendix

\section{Notation}\label{s:appe-notation}
In this section, we collect the main pieces of notation used in this paper.

\begin{table}[h]
    \centering
    \begin{tabular}{c|l}
        \toprule
        $T$ & Time horizon \\
        $K$ & Step of the grid over $[0, 1]$ \\
        $\cU$ & Upper triangle over $[0, 1]^2$ \\
        \bottomrule
        \addlinespace
        
        \toprule        
        \multicolumn{2}{c}{Market making} \\
        \midrule
        $\nP_t$ & Market price \\
        $\nZ_t$ & Taker's valuation \\
        $\nB_t$ & Bid presented by the learner \\
        $\nS_t$ & Ask presented by the learner \\
        \bottomrule
        \addlinespace

        \toprule        
        \multicolumn{2}{c}{First-price auction with unknown valuations} \\
        \midrule
        $\nV_t$ & Unknown valuation of the auctioned item \\
        $\nM_t$ & Highest competing bid in the auction \\
        $\nX_t$ & Bid presented by the learner \\
        \bottomrule
        \addlinespace
        
        \toprule        
        \multicolumn{2}{c}{Dynamic pricing with unknown costs} \\
        \midrule
        $\nC_t$ & Unknown cost of the item for sale \\
        $\nW_t$ & Buyer's valuation of the item \\
        $\nY_t$ & Price presented by the learner \\
        \bottomrule
    \end{tabular}
    \caption{Notation}
    \label{tab:notation}
\end{table}

\section{Additional Related Works}
\label{sec:additional}
There exists a vast literature that treats various trading-related tasks as specific stochastic control problems and solves them using techniques from stochastic control theory~\citep{shreve2004stochastic,shreve2005stochastic,cartea2015algorithmic,gatheral2011volatility}. However, most of that work assumes that the parameters of the underlying stochastic process are known or have been fitted to historical data in a previous ``calibration'' step. 
Fitting a distribution to data is challenging, especially in a market with hundreds of thousands of assets. Moreover, the distribution-fitting process is typically unaware of the downstream optimization problem the distribution is going to be fed into, and thus is unlikely to minimize generalization error on the downstream task. The field of online learning~\citep{10.5555/1137817,shalev2012online,hazan2016introduction,orabona2019modern} provides adversarial or distribution-free approaches for solving exactly the kinds of sequential decision making problems that are common in financial trading, and even though some recent research does apply these approaches to trading, they are still far from being widely adopted. 
An early contribution to this area is Cover's model of ``universal portfolios'' \citep{coverUniversalPortfolios1991}, where a problem of portfolio construction is solved in the case where asset returns are generated by an adversary. Cover showed that one can still achieve logarithmic regret with respect to the best ``constantly rebalanced portfolio'', that ensures the fraction of wealth allocated to each asset remains constant through time. A long line of work has built up on his model to handle transaction costs~\citep{blum1997universal} and side information~\citep{cover1996universal}, as well as optimizing the trade-off between regret and computation~\citep{helmbold1998line,borodin2000competitive,borodin2003can,zimmert2022pushing}. The problem of pricing derivatives such as options has also been tackled using online learning, where the prices are generated by an adversary as opposed to a geometric Brownian motion~\citep{demarzo2006online,abernethy2012minimax,abernethy2013hedge}.
However, most of the problems considered in the online learning framework are from the point of view of a liquidity seeker as opposed to a liquidity provider, such as a market maker. 

\paragraph{First-price auctions with unknown evaluations.}
When restricted solely to the option of buying stock, our problem can be viewed as a series of repeated first-price auctions with unknown valuations. This scenario has received prior attention within the context of regret minimization \citep{cesabianchiRoleTransparencyRepeated2024,pmlr-v157-achddou21a}. The present work leverages these existing results but applies them within a more intricate setting. 
\citet{cesabianchiRoleTransparencyRepeated2024} explored this problem with varying degrees of transparency, defined as the information revealed by the auctioneer, and provide a comprehensive characterization. 
In our framework, the level of transparency is higher than that encountered in the \textit{bandit} case, orthogonal to the \textit{semi-transparent} case, and less informative that the \textit{transparent} case; this positioning offers a novel perspective on the problem.
\citet{pmlr-v157-achddou21a} investigated the repeated first-price auction problem within a fixed stochastic environment with independence assumptions and provide instance-dependent bounds, we tackle a similar problem in Lemma \ref{lemma:upperbound-discretizedbandits} and present distribution-free guarantees.
Other types of feedback have been considered, e.g., the case where the maximum bids $(\nZ_t)_{t \in [T]}$ form an i.i.d.\ process and are observed only when the auction is lost \citep{han2024optimalnoregretlearningrepeated}, for the case where the private evaluations $(\nP_t)_{t \in [T]}$ are i.i.d.\ it is possible to achieve regret $\tilde O(\sqrt{T})$.

\paragraph{Dynamic pricing with unknown costs.}
When considering solely the option of selling stock, our problem aligns with the well-established field of dynamic pricing with unknown costs. This area boasts a rich body of research, prior research analyzed the setting in which the learner has $I$ items to sell to $N$ independent buyers and, with regularity assumptions of the underlying distributions, can achieve a $O((I\log T)^{\nicefrac{2}{3}})$ regret bound against an offline benchmark with knowledge of the buyer's distribution \citep{babaioffDynamicPricingLimited2015}. It is known that in the stochastic setting and, under light assumptions on the reward function, achieve regret $O(\sqrt{T\log{T}})$ \citep{kleinberg2003value}; this result has been expanded upon by considering discrete price distributions supported on a set of prices of unknown size $K$ and it has been shown to be possible to achieve regret of order $O(\sqrt{KT})$ \citep{cesa2019dynamic}. 
Our work on dynamic pricing focuses on a scenario where the learner maintains an unlimited inventory and must incur an unknown cost per trade before realizing any profit. 

\paragraph{Traditional finance.}
Traditionally, the finance literature first fits the parameters of a stochastic process to the market and then optimizes trading based on those parameters. For example, the Nobel prize-winning Black-Scholes-Merton formula~\citep{black1973pricing} assumes that the price of a stock follows a geometric Brownian motion with known volatility and prices an option on the stock by solving a Hamilton-Jacobi-Bellman equation to compute a dynamic trading strategy whose value at any time matches the value of the option. Many similar formulae have since been derived for more exotic derivatives and for more complex underlying stochastic processes---see~\citep{hull2017options} for an overview. A stochastic control approach has also been taken to solve the problem of ``optimal trade execution,'' which involves trading a large quantity of an asset over a specified period of time while minimizing market impact~\citep{almgren2001optimal,bertsimas1998optimal}, to compute a sequence of trades such that the total cost matches a pre-specified benchmark~\citep{konishi2002optimal}, and also for the problem of market making~\citep{cartea2015algorithmic}. Once again, the underlying stochastic process is assumed to be known.
The two-step approach of first fitting a distribution and then optimizing a function over the fitted distribution is very popular even for simpler questions. For example, the problem of ``portfolio optimization/construction'' deals with computing the optimal allocation of your wealth into various assets in order to maximize some notion of future utility. The celebrated Kelly criterion~\citep{shalunov2020calculated,thorp1975portfolio,breiman1961optimal,hakansson1975optimal,maclean2011kelly} recommends one to use the allocation that maximizes the expected log return, and the mean-variance theory~\citep{markowitz} (another Nobel prize-winning work) says one should maximize a sum of expected return and the variance of returns (scaled by some measure of your risk tolerance). However both assume that one knows the underlying distribution over returns.

\section{Missing details in the Proof of Theorem \ref{t:upper-bound-vin}}
\label{s:missing-details-VIN}

To guide the reader through the proof of \Cref{t:upper-bound-vin}, we split it in several steps.

First, consider the problem of repeated first-price auctions with unknown valuations. %
Since the feedback received at the end of each round allows to compute the utility of the learner, a natural strategy is to simply discretize the interval $[0,1]$ and run a bandit algorithm on the discretization.
The pseudo-code of this simple meta-algorithm can be found in \Cref{a:meta-discretized-bandits}. 
The following result shows its guarantees.

\begin{algorithm2e}[t]
    \Input{Number of arms $K\in \{2,3,\dots\}$, Algorithm $\cA$ for $K$-armed bandits}
    \Init{Let $q_k \ceq \frac{k-1}{K-1}$, for all $k\in[K]$}
    \For {time $t=1,2,\ldots$}{
        Let $I_t \in [K]$ be the output of $\cA$ at time $t$\;
        Post bid $\nX_t \ceq q_{ I_t }$ and observe feedback $\nV_t$ and $\one{ \nX_t \geq \nM_t }$\;
        Compute utility $(\nV_t - \nX_t) \cdot \one{ \nX_t \geq \nM_t }$ and feed it back to $\cA$\;
    }
    \caption{\db{}}
    \label{a:meta-discretized-bandits}
\end{algorithm2e}

\begin{lemma}\label{lemma:upperbound-discretizedbandits}
    In the repeated first-price auctions with unknown valuations problem,
    let $T \in \N$ be the time horizon and let $( \nV_t, \nM_t )_{t \in [T]}$ be the $[0,1]^2$-valued stochastic process representing the sequence of valuations and highest competing bids. 
    Assume that one of the two following conditions is satisfied:
    \begin{enumerate}
        \item \label{t:assOneSec} For each $t \in [T]$, the cumulative distribution function of $\nM_t$ is $L$-Lipschitz, for some $L>0$. 
        \item \label{t:assTwoSec} The process $( \nV_t , \nM_t )_{t \in [T]}$ is i.i.d.\ and, for each $t \in [T]$, the two random variables $\nV_t$ and $\nM_t$ are independent of each other.
    \end{enumerate}
    Then, for any $K \ge 2$ and any $K$-armed bandit algorithm $\cA$, letting $R^K_T$ be the regret of $\cA$ when the reward at any time $t\in[T]$ of any arm $k\in[K]$ is $(\nV_t-q_k)\one{q_k \ge \nM_t}$,
    the regret of~\Cref{a:meta-discretized-bandits} run with parameters $K$ and $\cA$ satisfies
    $
        R_T
    \le
        R_T^K + \frac{\tilde L + 1}{2(K-1)} T \;,
    $
    with $\tilde L = L$ (resp., $\tilde L = 1$) if \Cref{t:assOneSec} (resp., \Cref{t:assTwoSec}) holds. In particular, if $T \ge 2$, by choosing $K \coloneqq \lceil T^{\nicefrac{1}{3}} \rceil + 1$ and, as the underlying learning procedure $\cA$, an adapted version of Poly INF~\citep{audibert10a}, the regret of~\cref{a:meta-discretized-bandits} run with parameters $K$ and $\cA$ satisfies
    $
        R_T \le c \cdot T^{\nicefrac{2}{3}} \;,
    $ where $c \le L + 50$ (resp. $c \le 51$) if \Cref{t:assOne} (resp., \Cref{t:assTwo}) holds.
\end{lemma}

\begin{proof}
The regret $R_T^K$ of $\cA$, when the reward at any time $t\in[T]$ of an arm $k\in[K]$ is $(\nV_t-q_k)\one{q_k \ge \nM_t}$, is defined by
\[
    R_T^K
=
    \max_{k\in[K]} \E\lsb{ \sum_{t=1}^T ( \nV_t - q_k ) \cdot \one{ q_k \geq \nM_t } }
    -
    \E\lsb{ \sum_{t=1}^T ( \nV_t - q_{ I_t } ) \cdot \one{ q_{ I_t } \geq \nM_t } }.
\]
The regret $R_T$ of \cref{a:meta-discretized-bandits} when the unknown valuations are $\nV_1,\dots,\nV_T$ and the highest competing bids are $\nM_1,\dots,\nM_T$ is
\[
    R_T
=
    \sup_{x \in [0,1] } 
    \E\lsb{ \sum_{t=1}^T ( \nV_t - x ) \cdot \one{ x \geq \nM_t } }
    -
    \E\lsb{ \sum_{t=1}^T (\nV_t - q_{ I_t } ) \cdot \one{ q_{ I_t } \geq \nM_t } }.
\]
Hence $R_T = R_T^K + \delta_K $, where $\delta_K$ is the discretization error
\[
    \delta_K
\ceq
    \sup_{x \in [0,1] } 
    \E\lsb{ \sum_{t=1}^T ( \nV_t - x ) \cdot \one{ x \geq \nM_t } }
    -
    \max_{k\in[K]} \E\lsb{ \sum_{t=1}^T ( \nV_t - q_k ) \cdot \one{ q_k \geq \nM_t } }
    .
\]
We now proceed to bound the discretization error $\delta_K$ in the two cases:
\begin{enumerate}
    \item If the cumulative distribution function of $\nM_t$ is $L$-Lipschitz, for some $L > 0$, then for all $t \in [T]$ the expected utility is $\phi(q) = \E\bsb{ ( \nV_t - q ) \cdot \one{ q \geq \nM_t } }$ is $(L+1)$-Lipschitz; indeed, for all $q,q'\in[0,1]$, if (without loss of generality) $q>q'$, then
    \begin{multline*}
        \labs{ \phi(q) - \phi(q') }
    =
        \Babs{\E \bsb{(\nV_t-q)\one{\nM_t\le q} - (\nV_t-q')\one{\nM_t\le q'} } }
    \\ \qquad\qquad =
        \Babs{\E \bsb{(\nV_t-q) \one{q' < \nM_t \le q} + (q'-q) \one{\nM_t \le q'} }}
    \\ \le
        \brb{ \Pb[\nM_t \le q] - \Pb[\nM_t \le q'] } + (q-q')
    \le 
        (L+1) (q-q') .
    \end{multline*}
    Define the supremum $x^\star \in [0, 1]$ in the definition of $\delta_K$, which exists because $\phi$ is continuous and $[0, 1]$ is compact. Let $k^\star = \argmin_{k \in [K]} |q_{k} - x^\star|$ be the point in the grid closest to $x^\star$. Since
    \[
        \max_{k\in[K]} \E\lsb{ \sum_{t=1}^T ( \nV_t - q_k ) \cdot \one{ q_k \geq \nM_t } }
    \ge
        \E\lsb{ \sum_{t=1}^T ( \nV_t - q_{k^\star} ) \cdot \one{ q_{k^\star} \geq \nM_t } } \;,
    \]
    the function $\phi$ is $(L+1)$-Lipschitz, and $\labs{ x^\star - q_{k^\star} } \le \nicefrac{1}{2(K-1)}$ we obtain
    \[
        \delta_K
    \le
        \sum_{t=1}^T \Brb{ 
            \E\bsb{ ( \nV_t - x^\star ) \cdot \one{ x^\star \geq \nM_t } }
            -
            \E\bsb{ ( \nV_t - q_{k^\star} ) \cdot \one{ q_{k^\star} \geq \nM_t } }
        }
    \le
        \frac{L+1}{2(K-1)} T \;.
    \]

    \item If the process $(\nV_t,\nM_t)_{t \in [T]}$ is i.i.d.\ and for each $t \in [T]$ the two random variables $ \nV_t$ and $ \nM_t$ are independent of each other, then for all $t \in [T]$ the expected utility is $\psi(x) \coloneqq (\mu - x)F(x)$ where $\mu := \mathbb{E}[\nV_t]$ and $F(x) \coloneqq \Pb[\nM_t \leq x]$.
    Fix $\eta>0$.
    If $\sup_{x \in [0,1]}\psi(x) \le \eta$, set $k^\star = K$ and note that for each $x \in [0,1]$ we have that $\psi (x) \le \eta \le \psi(1) + \eta = \psi(q_{k^\star}) + \eta$, which means $\delta_K\leq\eta T$.
    Otherwise, let $x^\star_\eta$ be such that $\sup_{x \in [0,1]}\psi(x) - \psi(x^\star_\eta) \le \eta$ and notice that we can assume that $x^\star_\eta \le \mu$ because otherwise $\psi(x^\star_\eta) \le 0$, hence we would have been in the first case when $\sup_{x \in [0,1]}\psi(x) \le \eta$.
    The expected reward achieved by playing bid $x^\star_\eta$ can be controlled with any point $q_k$ on the grid such that $q_k \ge x^\star_\eta$ as
   \[
        \psi(x^\star_\eta) - \psi(q_k)
        = (\mu - x^\star_\eta)F(x_\eta^\star) - (\mu - q_k)F(q_k)
        \le (\mu - x^\star_\eta)F(q_k) - (\mu - q_k)F(q_k)
        \le q_k - x^\star_\eta
    \; .
    \]
    Call $k^\star \in [K]$ the index of the arm closest to $x_\eta^\star$ such that $q_{k^\star} \ge x^\star_\eta$, note that $q_{k^\star} - x^\star_\eta \le \nicefrac{1}{(K-1)}$. Thus
    \[
        \delta_K
    \le
        \eta T + \sum_{t=1}^T \lrb{ \psi(x^\star_\eta) - \psi(q_{k^\star}) }
    \le
        \eta T + \sum_{t=1}^T |x^\star_\eta - q_{k^\star}|
    \le
        \eta T + \frac{T}{K-1}\;.
    \]
    Given that $\eta$ was chosen arbitrarily, taking the limit as $\eta \rightarrow 0$, we have that
    \[
        \delta_K \le \frac{T}{K-1}\;.
    \]
\end{enumerate}
Define $\tilde L \coloneqq L$ in the first case and $\tilde L \coloneqq 1$ in the second case.
 Next, pick the Poly INF algorithm~\citep{audibert10a} as the underlying learning procedure $\cA$ and apply the appropriate rescaling of the utilities $x \mapsto \frac{x+1}{2}$, which is necessary because the utility yields values in $[-1, 1]$, while Poly INF was designed for rewards in $[0, 1]$, this costs a multiplicative factor of 2 in the regret guarantee.
The regret $R_T$ of \cref{a:meta-discretized-bandits} can be upper bounded by
\[
    R_T
    \le R^K_T + \frac{\tilde L}{K - 1} \cdot T
    \le 50 \cdot T^{\nicefrac{2}{3}} + \frac{\tilde L}{K - 1} \cdot T
    \le (50 + \tilde L) \cdot T^{\nicefrac{2}{3}}
\]
whenever $K = \lceil T^{\nicefrac{1}{3}} \rceil + 1$ and the second inequality comes from the regret guarantees of the rescaled version of Poly INF~\cite[Theorem 11]{audibert10a}.
\end{proof}

A completely analogous result can be proved for \Cref{a:meta-discretized-bandits-dp} for the problem of dynamic pricing with unknown costs.
\begin{algorithm2e}[t]
    \Input{Number of arms $K\in \{2,3,\dots\}$, Algorithm $\cA$ for $K$-armed bandits}
    \Init{Let $q_k \ceq \frac{k-1}{K-1}$, for all $k\in[K]$}
    \For{time $t=1,2,\ldots$}{
        Let $I_t \in [K]$ be the output of $\cA$ at time $t$\;
        Post price $\nY_t \ceq q_{ I_t }$ and observe feedback $\nC_t$ and $\one{ \nY_t < \nW_t }$\;
        Compute utility $(\nY_t - \nC_t) \cdot \one{ \nY_t < \nW_t }$ and feed it back to $\cA$\;
    }
    \caption{\dbdp{}\label{a:meta-discretized-bandits-dp}}
\end{algorithm2e}

\begin{lemma}\label{lemma:upperbound-discretizedbandits-dp}
    In the repeated dynamic pricing with unknown costs problem,
    let $T \in \N$ be the time horizon and let $( \nC_t, \nW_t )_{t \in [T]}$ be the $[0,1]^2$-valued stochastic process representing the sequence of costs and buyer's valuations. 
    Assume that one of the two following conditions is satisfied:
    \begin{enumerate}
        \item \label{t:assOne}
        For each $t \in [T]$, the cumulative distribution function of $\nW_t$ is $L$-Lipschitz, for some $L>0$. 
        \item \label{t:assTwo}
        The process $( \nC_t , \nW_t )_{t \in [T]}$ is i.i.d.\ and, for each $t \in [T]$, the two random variables $\nC_t$ and $\nW_t$ are independent of each other.
    \end{enumerate}
    Then, for any $K \ge 2$ and any $K$-armed bandit algorithm $\cA$, letting $R^K_T$ be the regret of $\cA$ when the reward at any time $t\in[T]$ of any arm $k\in[K]$ is $(q_k - \nC_t)\one{q_k < \nW_t}$,
    the regret of~\Cref{a:meta-discretized-bandits-dp} run with parameters $K$ and $\cA$ satisfies
    \[
        R_T
    \le
        R_T^K + \frac{\tilde L + 1}{2(K-1)} T \;,
    \]
    with $\tilde L = L$ (resp., $\tilde L = 1$) if \Cref{t:assOne} (resp., \Cref{t:assTwo}) holds. In particular, if $T \ge 2$, by choosing $K \coloneqq \lceil T^{\nicefrac{1}{3}} \rceil + 1$ and, as the underlying learning procedure $\cA$, an adapted version of Poly INF~\citep{audibert10a}, the regret of~\cref{a:meta-discretized-bandits-dp} run with parameters $K$ and $\cA$ satisfies
    \[
        R_T \le c \cdot T^{\nicefrac{2}{3}} \;,
    \] where $c \le L + 50$ (resp. $c \le 51$) if \Cref{t:assOne} (resp., \Cref{t:assTwo}) holds.
\end{lemma}
\begin{proof}
    As in Lemma \ref{lemma:upperbound-discretizedbandits}, we can bound $R_T$ for the two cases by controlling the discretization error $\delta_K$ on the grid via Lipschitzness (Item 1) or leveraging the structure of the expected utility around the maximum (Item 2). The bound on the regret $\cR_T$ of $\cA$ follows from the same guarantees on Poly INF~\citep{audibert10a}.
\end{proof}

Next, we use Lemma \ref{lemma:upperbound-discretizedbandits} and Lemma \ref{lemma:upperbound-discretizedbandits-dp} to show that the regret of an appropriately tuned \vin{} is bounded by $\cO(T^{\nicefrac{2}{3}})$.
This implies is the second part of \Cref{t:upper-bound-vin}, fully concluding its proof.

\begin{lemma}\label{t:upper-T23-VIN}
    In the online market-making problem,
    let $T \in \N$ be the time horizon and let $(\nP_t,\nZ_t)_{t \in [T]}$ be the $[0,1]^2$-valued stochastic process representing the sequence of market prices and takers' valuations.
    Assume that one of the two following conditions is satisfied:
    \begin{enumerate}
        \item \label{t:assOnePrimeAppe}
        For each $t \in [T]$, the cumulative distribution function of $\nZ_t$ is $L$-Lipschitz, for some $L>0$. 
        \item \label{t:assTwoPrimeAppe}
        The process $( \nP_t , \nZ_t )_{t \in [T]}$ is i.i.d.\ and, for each $t \in [T]$, the two random variables $\nP_t$ and $\nZ_t$ are independent of each other.
    \end{enumerate}
    Then, let $K \coloneqq \lceil T^{\nicefrac{1}{3}} \rceil + 1$, let $\cA^{\mathrm{fpa}}$ be the instance of \cref{a:meta-discretized-bandits} that uses Poly INF as described in Lemma \ref{lemma:upperbound-discretizedbandits}, and $\cA^{\mathrm{dp}}$ be the instance of \cref{a:meta-discretized-bandits-dp} that uses Poly INF as described in Lemma \ref{lemma:upperbound-discretizedbandits-dp}. Then the regret of \vin{} (\Cref{a:meta-vin}) run with parameters $\cA^{\mathrm{fpa}}$ and $\cA^{\mathrm{dp}}$ satisfies
    $
        R_T
    \le
        c T^{\nicefrac{2}{3}}
    $
    with $c \le 2L+100$ (resp., $c \le 102$) if \Cref{t:assOnePrimeAppe} (resp., \Cref{t:assTwoPrimeAppe}) holds.
\end{lemma}
\begin{proof}
    By \Cref{a:meta-vin}, the regret of \vin{} is upper bounded by the regret of a repeated first-price auctions problem with unknown valuations plus a dynamic pricing problem with unknown costs. Plugging in the bounds from Lemma \ref{lemma:upperbound-discretizedbandits} and Lemma \ref{lemma:upperbound-discretizedbandits-dp}, we get the required result.
\end{proof}

\section{Proof of Theorem \ref{t:lower-bound-lip-iv}}
\label{s:kl-bounds}

The following is the proof of our $T^{2/3}$ lower bound (\Cref{t:lower-bound-lip-iv}). 

{
\begin{proof}
\renewcommand{\cA}{\mathscr{A}}
\renewcommand{\cS}{\mathcal{A}}
Fix $T \ge 42$.
We define the following constants that will be used in the proof.
\begin{align*}
    K & \coloneqq \bigl\lceil T^{\nicefrac{1}{3}} \bigr\rceil 
& 
    \e_K & \coloneqq \nicefrac{1}{16 K} 
&
     \forall k\in[K], \ r_{K}^k & \coloneqq \nicefrac{3}{16} + \lrb{k - \nicefrac{1}{2}}\e_K
\\ 
    p_{\textbf{left}} & \coloneqq \nicefrac{3}{16} 
& 
    p_{\textbf{right}} & \coloneqq \nicefrac{1}{4} 
& 
    p_{\textbf{exploit}} & \coloneqq \nicefrac{3}{4}
\\
    c_{\textbf{plat}} & \coloneqq \nicefrac{1}{32}
&
    c_{\textbf{spike}} & \coloneqq \nicefrac{1}{72}
&
\end{align*}

Define the density
\[
   f \colon [0,1] \to [0,\infty)
    \;, \qquad
    x \mapsto 
    \frac{8}{9} \cdot \I_{\lsb{0, \frac{3}{16}}}(x) + 
    \frac{1}{8} \frac{1}{\lrb{ \frac{15}{16} - x}^2 } \cdot 
        \I_{ \left( \frac{3}{16}, \frac{3}{4} \right] } (x) +
    \frac{8}{3} \cdot \I_{ \left( \frac{3}{4}, \frac{7}{8} \right] } (x)
    \;,
\]
\begin{figure}[t]
    \centering
    \includegraphics[page=3]{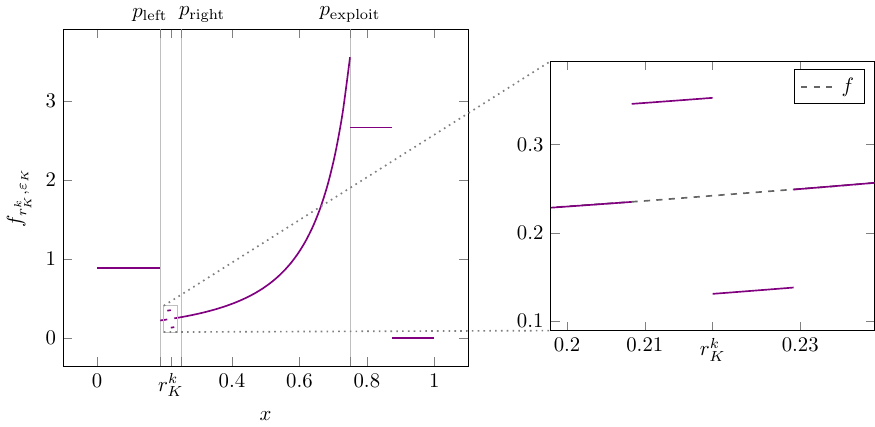}
    \includegraphics[page=4]{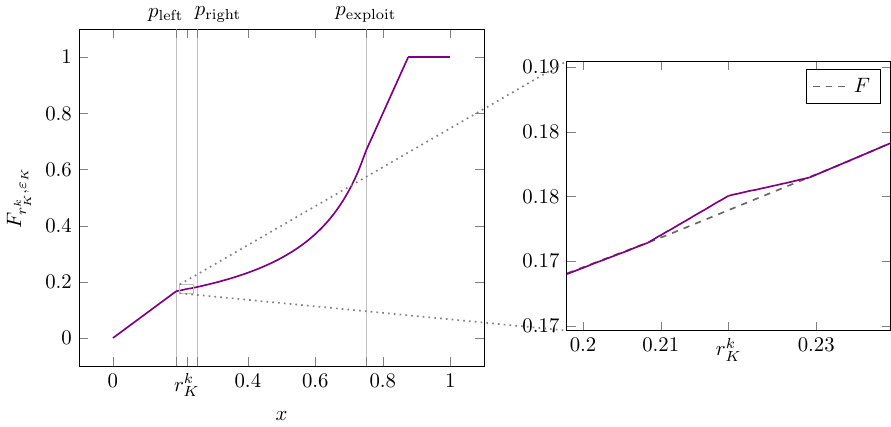}
    \caption{Above, the perturbed density $f_{r^k_K,\e_K}(x)$ and below, the perturbed \textit{cumulative} density $F_{r^k_K,\e_K}(x)$ with $K = \TK$, $k = \Tk$, the highlighted portions show the effects of the perturbation $\Lambda_{r^k_K, \e_K}$ on the density function $f(x)$ and on the cumulative distribution function $F(x)$ respectively. The dotted lines represent the base functions.}
    \label{fig:pertburbed-density}
\end{figure}
so that the corresponding cumulative distribution function $F$ satisfies, for each $x \in [0,1]$,
\begin{equation}\label{eq:lowerbound-cdf}
    F(x)
=
    \frac{8}{9} x \cdot 
    \I_{ \lsb{ 0, \frac{3}{16} } } (x)
    + \frac{1}{8} \frac{1}{\frac{15}{16} - x} \cdot
    \I_{ \left( \frac{3}{16},\frac{3}{4} \right] } (x) +
    \frac{8}{3} \lrb{ x - \frac{1}{2} } \cdot
    \I_{ \left( \frac{3}{4}, \frac{7}{8} \right] } (x) +
    \I_{ \left( \frac{7}{8}, 1 \right] } (x) \; .
\end{equation}
We define a family of perturbations parameterized by the set $\Xi \coloneqq \{ (r,\e) \in [\frac{3}{16},\frac{11}{16}] \times [0,1] \mid  \frac{3}{16} \le r-\frac{\e}{2} \le r+\frac{\e}{2} \le \frac{11}{16}\}$;
for each $(r,\e) \in \Xi$, define
$
    g_{r,\e}
\coloneqq
    \frac{1}{9}\cdot \I_{[r-\frac{\e}{2},r]} - \frac{1}{9} \cdot \I_{(r,r+\frac{\e}{2}]}
$
and $f_{r,\e} \coloneq f + g_{r,\e}$ (see \Cref{fig:pertburbed-density}).
Notice that for each $(r,\e) \in  \Xi $ the function $f_{r, \e}$ is still a density function whose corresponding cumulative distribution function $F_{r,\e}$ satisfies, for each $x \in [0,1]$,
\[
    F_{r,\e}(x)
=
    F(x)
    +
    \frac{\e}{18}\Lambda_{r,\e}(x)\;,
\]
where $\Lambda_{r,\e}$ is the tent function of height $1$ and width $\e$ centered in $r$, i.e., the function defined, for each $x \in \bbR$, by
\[
    \Lambda_{r,\e}(x)
\coloneqq
    \lrb{ 1 - \frac{2}{\e}(r-x) } \cdot \I_{\lsb{r-\frac{\e}{2},r}}(x)
    +
    \lrb{ 1 - \frac{2}{\e} (x-r) } \cdot \I_{\left(r,r+\frac{\e}{2}\right]}(x) \;.
\]
Note that $F_{r,\e}$ is $4$-Lipschitz;
indeed, for each $(\nb, \ns) \in \cU$,
\begin{multline*}
    \labs{ F_{r,\e}(\ns) - F_{r,\e}(\nb) }
    = \int^\ns_\nb f_{r, \e}(x) \, \dif x
    = \max_{c \in [\nb,\ns]}f_{r, \e}(c) (\ns - \nb)
\\
    \le \left( \max_{c' \in [0, 1]} f(c') + \frac{1}{9} \right) (\ns - \nb)
    = 4 (\ns - \nb) \, ,
\end{multline*}
where $f$ is maximized in $\nicefrac{3}{4}$.
Consider an independent family $\{\nP_t,\nZ_t,\nZ_{r,\e,t}\}_{t \in \N, (r,\e)\in \Xi}$ such that for each $t \in \N$ the distribution $\mu$ of $\nP_t$ is a uniform on $\lsb{\frac{7}{8},1}$ (therefore, $M_t$ admits an $8$-Lipschitz cumulative distribution function), for each $t \in \N$ the distribution $\nu$ of $\nZ_t$ has $f$ as density, while for each $(r,\e) \in \Xi$ and each $t \in \N$ the distribution $\nu_{r,\e}$ of $\nZ_{r,\e,t}$ has $f_{r,\e}$ as density.
Notice that for each $k \in [K]$ we have that $(r_K^k,\e_K) \in \Xi$.
Now, partition $\cU$ in the following regions (see \Cref{fig:density-regions-lower-bound} for a not-to-scale illustration).
\begin{align*}
    R^{\textbf{left}}_1
&\coloneqq
    \lcb{ (\nb, \ns) \in \cU \mid (\nb \le p_{\textbf{left}}) \land \lrb{p_{\textbf{right}} - \e_K \le \ns \le p_{\textbf{right}}}}
\\
    R^{\textbf{left}}_{2}
&\coloneqq
    \lcb{ (\nb, \ns) \in \cU \mid (\nb \le p_{\textbf{left}}) \land \lrb{p_{\textbf{right}} - 2 \e_K \le \ns < p_{\textbf{right}} -\e_K}}
\\
&\vdots
\\
R^{\textbf{left}}_{K-1}
&\coloneqq
    \lcb{ (\nb, \ns) \in \cU \mid (\nb \le p_{\textbf{left}}) \land \lrb{p_{\textbf{left}} + \e_K \le \ns < p_{\textbf{left}} + 2\e_K}}
\\
    R^{\textbf{left}}_K
&\coloneqq
    \lcb{ (\nb, \ns) \in \cU \mid (\nb \le p_{\textbf{left}}) \land \lrb{p_{\textbf{left}} \le \ns < p_{\textbf{left}} + \e_K}}
\end{align*}
and
\begin{align*}
    R^{\textbf{top}}_1
&\coloneqq
    \lcb{ (\nb, \ns) \in \cU \mid (p_{\textbf{left}} \le \nb \le p_{\textbf{left}} + \e_K) \land \lrb{p_{\textbf{right}} \le \ns \le p_{\textbf{exploit}}}}
\\
    R^{\textbf{top}}_2
&\coloneqq
    \lcb{ (\nb, \ns) \in \cU \mid (p_{\textbf{left}} + \e_K < \nb \le p_{\textbf{left}} + 2 \e_K) \land \lrb{p_{\textbf{right}} \le \ns \le p_{\textbf{exploit}}}}
\\
&\vdots
\\
    R^{\textbf{top}}_{K-1}
&\coloneqq
    \lcb{ (\nb, \ns) \in \cU \mid (p_{\textbf{right}} - 2 \e_K < \nb \le p_{\textbf{right}} - \e_K) \land \lrb{p_{\textbf{right}} \le \ns \le p_{\textbf{exploit}}}}
\\
    R^{\textbf{top}}_K
&\coloneqq
    \lcb{ (\nb, \ns) \in \cU \mid (p_{\textbf{right}} - \e_K < \nb \le p_{\textbf{right}} ) \land \lrb{p_{\textbf{right}} \le \ns \le p_{\textbf{exploit}}}}
\end{align*}
and, for each $i,j\in [K]$ such that $i+j \le K$
\begin{align*}
    R^{\textbf{square}}_{i,j}
    \coloneqq
    \{
        (\nb, \ns) \in \cU
        & \mid
        (p_{\textbf{left}} + (i - 1) \e_K < \nb \le p_{\textbf{left}} + i \e_K) \land \\
        & \quad 
        (p_{\textbf{right}} - j \e_K < \ns \le p_{\textbf{right}} - (j - 1) \e_K)
    \}
\end{align*}
and, for each $k \in [K]$
\begin{align*}
    R^{\textbf{triangle}}_{k}
    \coloneqq
    \{
        (\nb, \ns) \in \cU
        & \mid
        (p_{\textbf{left}} + (k - 1) \e_K < \nb \le p_{\textbf{left}} + k \e_K) \land \\
        & \quad 
        (p_{\textbf{left}} + (k - 1) \e_K < \ns \le p_{\textbf{left}} - k \e_K)
    \}
\end{align*}
and
\begin{align*}
    R^{\textbf{exploit}}_{1}
    &\coloneqq
    \{
        (\nb, \ns) \in \cU
        \mid
        \lcb{ (\nb, \ns) \in \cU \mid (p_{\textbf{left}} \le \nb \le p_{\textbf{left}} + \e_K) \land \lrb{p_{\textbf{exploit}} < \ns}}
    \} \\
    R^{\textbf{exploit}}_{2}
    &\coloneqq
    \{
        (\nb, \ns) \in \cU
        \mid
        \lcb{ (\nb, \ns) \in \cU \mid (p_{\textbf{left}} + \e_K < \nb \le p_{\textbf{left}} + 2\e_K) \land \lrb{p_{\textbf{exploit}} < \ns}}
    \} \\
    \vdots \\
    R^{\textbf{exploit}}_{K - 1}
    &\coloneqq
    \{
        (\nb, \ns) \in \cU
        \mid
        \lcb{ (\nb, \ns) \in \cU \mid (p_{\textbf{right}} - 2 \e_K < \nb \le p_{\textbf{right}} - \e_K) \land \lrb{p_{\textbf{exploit}} < \ns}}
    \} \\
        R^{\textbf{exploit}}_{K}
    &\coloneqq
    \{
        (\nb, \ns) \in \cU
        \mid
        \lcb{ (\nb, \ns) \in \cU \mid (p_{\textbf{right}} - \e_K < \nb \le p_{\textbf{right}}) \land \lrb{p_{\textbf{exploit}} < \ns}}
    \}
\end{align*}
Let $R^{\textbf{white}}$ be the part of $\cU$ not covered by the union of the previous regions and define
\[
    R^{\textbf{exploit}}
    \coloneqq R^{\textbf{exploit}}_1 \cup \cdots \cup R^{\textbf{exploit}}_K \; ,
    \qquad
    R^{\textbf{explore}}
    \coloneqq \cU \backslash (R^{\textbf{white}} \cup R^{\textbf{exploit}}) \; .
\]

Notice that, for each $k \in [K]$, each $t \in \N$, and each $(\nb, \ns) \in \cU$, given that $\Pb [ \nZ_{r_K^k,\e_K,t} \le \nP_t ] = 1$, it holds that
\begin{align*}
&
    \E\bsb{\util(\nb,\ns,\nP_t,\nZ_{r_K^k,\e_K,t})}
\le
    \E\bsb{\util(\nb,1,\nP_t,\nZ_{r_K^k,\e_K,t})}
\\ &
=
    \E\bsb{(\nP_t - \nb) \one{\nb > \nZ_{r_K^k,\e_K,t}}}
=
    \lrb{\frac{15}{16} - \nb}F_{r_K^k,\e_K}(\nb)
\\ &
=
    \frac{8}{9}\nb\lrb{\frac{15}{16}-b} \cdot \I_{\left[0, \frac{3}{16}\right]}(\nb) 
+   
    \frac{1}{8}\cdot \I_{\left(\frac{3}{16},\frac{3}{4}\right]}(\nb)
+
    \frac{\e_K}{18} \lrb{\frac{15}{16} - \nb} \cdot \Lambda_{r_K^k,\e_K}(\nb)
\\ & \qquad
+
    \frac{8}{3} \lrb{\nb - \frac{1}{2}} \lrb{\frac{15}{16} - \nb} \cdot \I_{(\frac{3}{4},\frac{7}{8}]}(\nb)
+
    \lrb{\frac{15}{16} - \nb} \cdot \I_{\left(\frac{7}{8}, 1\right]}(\nb)\;,
\end{align*}
from which it follows that:
\begin{itemize}
    \itemm The pair of prices with highest expected utility is $(r_K^k, 1)$ and
    $$
        \max_{(\nb, \ns) \in \cU} \E\bsb{\util(\nb,\ns,\nP_t, \nZ_{r_K^k,\e_K,t})} = \E\lsb{\util\lrb{r_K^k,1,\nP_t,\nZ_{r_K^k,\e_K,t}}} \ge \frac{1}{8} + c_{\textbf{spike}} \cdot \e_K \;.
    $$
    \itemm The maximum expected utility when $\nb$ is not in the perturbation is
    $$
        \max_{
            (\nb,\ns) \in \cU,
            \nb \notin [r_K^k-\frac{\e_K}{2}, r_K^k+\frac{\e_K}{2}]
        }
        \E \lsb{ \util(\nb, \ns, \nP_t, \nZ_{r_K^k, \e_K,t}) }
        = \frac{1}{8}
        \; .
    $$
    \itemm The expected utility in the exploration region is upper bounded by
    $$
        \max_{(\nb, \ns) \in R^{\textbf{explore}}}
        \E\lsb{\util(\nb,\ns,\nP_t,\nZ_{r_K^k,\e_K,t})}
        \le \E\lsb{\util(r_K^k, p_\textbf{exploit}, \nP_t, \nZ_{r_K^k,\e_K,t})}
        \le \frac{1}{8} - c_{\textbf{plat}}
        \;.
    $$
\end{itemize}
Also note that:
\begin{itemize}
    \itemm For each $x \in [0,1]$, if $x \notin [r_K^k-\nicefrac{\e_K}{2},r_K^k+\nicefrac{\e_K}{2}]$, then $\Pb[\nZ_{r_K^k,\e_K,t}<x] = F(x)$.
\end{itemize}
Crucially, given that the setting is stochastic, without loss of generality, we can consider only deterministic algorithms.
Fix a deterministic algorithm $(\cA_t)_{t \in \N} \coloneqq (\ncB_t,\ncS_t)_{t \in \N}$, i.e., a sequence of functions such that for each $t \in \N$ we have that $\cA_t = (\ncB_{t},\ncS_{t}) \colon ([0,1] \times \{0,1\} \times \{0,1\})^{t-1} \to \cU$, with the understanding that $\cA_1 = (\ncB_1,\ncS_1) \in \cU$.
For each $k \in [K]$, let $(\nB_t^k,\nS_t^k)_{t \in [T]}$ be the prices posted by the algorithm when the underlying instance is $(\nP_t,\nZ_{\e_K,r_K^k,t})_{t \in [T]}$, i.e., let $(\nB_1^k,\nS_1^k) \coloneqq \cA_1$ and for each $t \in [T]$ with $t \ge 2$ let $(\nB_t^k,\nS_t^k) \coloneqq \cA_t(\nP_1,\one{\nZ_{\e_K,r_K^k,1}\le\nB_1^k},\one{\nZ_{\e_K,r_K^k,1}>\nS_1^k},\dots,\nP_{t-1},\one{\nZ_{\e_K,r_K^k,{t-1}}\le B_{t-1}^k},\one{\nZ_{\e_K,r_K^k,{t-1}}> \nS_{t-1}^k})$.
Analogously, let $(\nB_t,\nS_t)_{t \in [T]}$ be the prices posted by the algorithm when the underlying instance is $(\nP_t,\nZ_t)_{t \in [T]}$.
For each $k \in [K]$ and each $t \in [T]$, let also $\nW_t^k \coloneqq (\one{\nZ_{\e_K,r_K^k,t} \le \nB_t^k},\one{\nZ_{\e_K,r_K^k,t}> \nS_t^k})$ and $\nW_t \coloneqq (\one{\nZ_t \le \nB_t},\one{\nZ_t> \nS_t})$.

\begin{figure}[t]
    \centering
    \includegraphics[page=5]{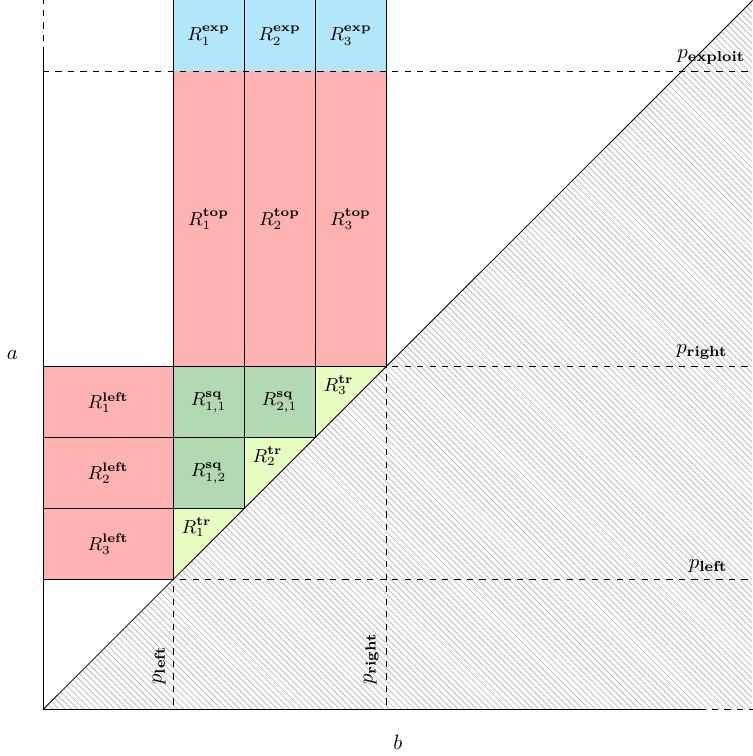}
    \caption{(Not-to-scale) Division of the upper triangle $\cU$ in regions with $K = \RK$, to illustrate the relative positions of the exploitation regions (in blue) and the exploration regions (in red, green and lime).}
    \label{fig:density-regions-lower-bound}
\end{figure}

Define the following auxiliary random variables.
For each $i,k \in [K]$, define
\[
    N^{\textbf{left}}_{i,k}(t) \coloneqq \sum_{s=1}^t \I\lcb{ \cA_s(\nP_1,\nW_1^k,\dots,\nP_{s-1},\nW_{s-1}^k) \in R^{\textbf{left}}_i }
\]
and for each $i \in [K]$, define
\[
    N^{\textbf{left}}_{i}(t) \coloneqq \sum_{s=1}^t \I\lcb{ \cA_s(\nP_1,\nW_1,\dots,\nP_{s-1},\nW_{s-1}) \in R^{\textbf{left}}_i }\;.
\]
Analogously, for each $i,k \in [K]$, define
\[
    N^{\textbf{top}}_{i,k}(t) \coloneqq \sum_{s=1}^t \I\lcb{ \cA_s(\nP_1,\nW_1^k,\dots,\nP_{s-1},\nW_{s-1}^k) \in R^{\textbf{top}}_i }
\]
and for each $i \in [K]$, define
\[
    N^{\textbf{top}}_{i}(t) \coloneqq \sum_{s=1}^t \I\lcb{ \cA_s(\nP_1,\nW_1,\dots,\nP_{s-1},\nW_{s-1}) \in R^{\textbf{top}}_i }\;.
\]
Also, for each $i,j,k \in [K]$ with $i+j\le K$, define
\[
    N^{\textbf{square}}_{i,j,k}(t) \coloneqq \sum_{s=1}^t \I\lcb{ \cA_s(\nP_1,\nW_1^k,\dots,\nP_{s-1},\nW_{s-1}^k) \in  R^{\textbf{square}}_{i,j} }
\]
and for each $i,j \in [K]$ with $i+j\le K$, define
\[
    N^{\textbf{square}}_{i,j}(t) \coloneqq \sum_{s=1}^t \I\lcb{ \cA_s(\nP_1,\nW_1,\dots,\nP_{s-1},\nW_{s-1}) \in R^{\textbf{square}}_{i,j}  }\;.
\]
Then, for each $i,k \in [K]$, define
\[
    N^{\textbf{triangle}}_{i,k}(t) \coloneqq \sum_{s=1}^t \I\lcb{ \cA_s(\nP_1,\nW_1^k,\dots,\nP_{s-1},\nW_{s-1}^k) \in R^{\textbf{triangle}}_{i}  }
\]
and for each $i \in [K]$, define
\[
    N^{\textbf{triangle}}_{i}(t) \coloneqq \sum_{s=1}^t \I\lcb{ \cA_s(\nP_1,\nW_1,\dots,\nP_{s-1},\nW_{s-1}) \in R^{\textbf{triangle}}_{i}  } \;.
\]
Moreover, for each $i,k \in [K]$, define
\[
    N^{\textbf{exploit}}_{i,k}(t) \coloneqq \sum_{s=1}^t \I\lcb{ \cA_s(\nP_1,\nW_1^k,\dots,\nP_{s-1},\nW_{s-1}^k) \in R^{\textbf{exploit}}_i }
\]
and for each $i \in [K]$, define
\[
    N^{\textbf{exploit}}_{i}(t) \coloneqq \sum_{s=1}^t \I\lcb{ \cA_s(\nP_1,\nW_1,\dots,\nP_{s-1},\nW_{s-1}) \in R^{\textbf{exploit}}_i }\;.
\]
Finally, for each $k \in [K]$, define
\[
    N^{\textbf{white}}_{k}(t) \coloneqq \sum_{s=1}^t \I\lcb{ \cA_s(\nP_1,\nW_1^k,\dots,\nP_{s-1},\nW_{s-1}^k) \in R^{\textbf{white}} }
\]
and for each $i \in [K]$, define
\[
    N^{\textbf{white}}(t) \coloneqq \sum_{s=1}^t \I\lcb{ \cA_s(\nP_1,\nW_1,\dots,\nP_{s-1},\nW_{s-1}) \in R^{\textbf{white}} }\;.
\]
Let $R^k_T$ be the regret of the algorithm $\cA$ up to the time horizon $T$ when the underlying instance is $(\nZ_t^k,\nP_t)_{t \in \N}$ and let $R_T$ be the regret of the algorithm $\cA$ up to the time horizon $T$ when the underlying instance is $(\nZ_t,\nP_t)_{t \in \N}$.
Start by considering
\begin{align*}
    \frac{1}{K} \sum_{k \in [K]} R^k_T
&\ge
    \frac{1}{K} \sum_{k \in [K]} \lrb{ c_{\textbf{spike}} \cdot \e_K \cdot \Brb{ T - \E \bsb{ N^{\textbf{exploit}}_{k,k}(T) } } }
\eqqcolon
    (\square)\;.
\end{align*}
We recall that if $(\cX,\cF)$ is a measurable space and $X$ is a $\cX$-valued random variable, we denote by $\Pb_X$ the push-forward probability measure of $\Pb$ induced by $X$ on $\cX$, i.e., $\Pb_X [E] \coloneqq \Pb [X \in E] $, for any $E \in \cF$.
Also, with $\lno{\cdot}_{\text{TV}}$ we denote the total variation norm for measures, and for any two probability measures $\mathbb{Q},\mathbb{Q}'$ defined on the same sample space, we denote their Kullback-Leibler divergence by $\cD_{\text{KL}}(\mathbb{Q}, \mathbb{Q}')$. 
Now, for each $k \in [K]$, using Pinsker's inequality \cite[Lemma 2.5]{Tsybakov2008} that upper bounds the difference in the total variation $\lno{\cdot}_{\text{TV}}$ of two probability measures using their Kullback-Leibler divergence $\cD_{\text{KL}}$, we have that 
\begin{align*}
    \labs{\E[N^{\textbf{exploit}}_{k,k}(T)] - \E[N_k^{\textbf{exploit}}(T)]}
&\le
    \sum_{t=1}^T\labs{\Pb[(\nB_t^k,\nS_t^k) \in R^{\textbf{exploit}}_k] - \Pb[(\nB_t,\nS_t) \in R^{\textbf{exploit}}_k]}
\\
&\le
    \sum_{t=1}^{T-1} \lno{ \Pb_{(\nP_1,\nW_1,\dots,\nP_t,\nW_t)} - \Pb_{(\nP_1,\nW_1^k,\dots,\nP_t,\nW_t^k)} }_{\text{TV}}
\\
&\le
    \sum_{t=1}^{T-1} \sqrt{\frac{1}{2} \cD_{\text{KL}}\lrb{\Pb_{(\nP_1,\nW_1,\dots,\nP_t,\nW_t)},\Pb_{(\nP_1,\nW_1^k,\dots,\nP_t,\nW_t^k)}}}
\eqqcolon
    (\star_k)\;.
\end{align*}
Now, for each $t \in [T-1]$, each $\np_1,\dots,\np_t \in [0,1]$, and each $\nw_1,\dots,\nw_t \in \{0,1\}^{2}$, defining $\ns \coloneqq \ncS_{t+1}(\np_1,\nw_1,\dots,\np_t,\nw_t)$ and $\nb \coloneqq \ncB_{t+1}(\np_1,\nw_1,\dots,\np_t,\nw_t)$, we have
\begin{align*}
&
    \cD_{\text{KL}}
    \big(
        \Pb_{(\nP_{t+1},W_{t+1}) \mid (\nP_1,\nW_1,\dots,\nP_{t},W_{t}) = (\np_1,\nw_1,\dots,\np_{t},w_{t}))},
\\ & \qquad\qquad
        \Pb_{(\nP_{t+1},W_{t+1}^k) \mid (\nP_1,\nW_1^k,\dots,\nP_{t},W_{t}^k) = (\np_1,\nw_1,\dots,\np_{t},w_{t}))}
    \big)
\\ &=
    \cD_{\text{KL}}
    \lrb{
        \Pb_{(\one{\nZ_{t+1}\le b}, \one{\ns < \nZ_{t+1}}) },
        \Pb_{(\one{\nZ_{t+1}^k \le b}, \one{\ns < \nZ_{t+1}^k})}
    }
\\ & \qquad
    \cdot
    \bigg( 
        \one{ (\ns,\nb) \in R^{\textbf{left}}_k  } 
        + \sum_{j=1}^{k-1} \one{ (\ns,\nb) \in R^{\textbf{square}}_{k,j} } 
        + \one{ (\ns,\nb) \in R^{\textbf{triangle}}_{k} } 
\\ & \qquad\qquad
        + \sum_{i=1}^{k-1} \one{ (\ns,\nb) \in R^{\textbf{square}}_{i,k} }
        + \one{ (\ns,\nb) \in R^{\textbf{top}}_k  }
        + \one{ (\ns,\nb) \in R^{\textbf{exploit}}_k  } 
    \bigg)
\end{align*}
A direct verification (see Lemma \ref{lemma:kl-upperbounds}) shows that, letting $c_2 \ceq \nicefrac{65}{9}$, for each $k \in [K]$, and $i,j \in [K-1]$, if $(\ns,\nb) \in R^{\textbf{left}}_k \cup R^{\textbf{square}}_{k,j} \cup R^{\textbf{triangle}}_{k} \cup R^{\textbf{square}}_{i,k} \cup R^{\textbf{top}}_{k}$, then
\[
    \cD_{\text{KL}}
    \lrb
    {
        \Pb_{(\one{\nZ_{t+1} \le \nb}, \one{ \ns < \nZ_{t+1}}) },\Pb_{(\one{\nZ_{t+1}^k \le \nb}, \one{ \ns < \nZ_{t+1}^k})}
    }
\le
    c_2 \cdot \e_K \;,
\]
likewise (see, again, Lemma \ref{lemma:kl-upperbounds}), letting $\nc_1\ceq \nicefrac{2}{81}$ such that, for each $k \in [K]$, if $(\ns,\nb) \in R^{\textbf{exploit}}_k$, then
\[
     \cD_{\text{KL}}
    \lrb
    {
        \Pb_{(\one{\nZ_{t+1} \le \nb}, \one{\ns < \nZ_{t+1}}) },\Pb_{(\one{\nZ_{t+1}^k \le \nb}, \one{\ns < Z_{t+1}^k})}
    }
\le
    \nc_1 \cdot \e_K^2\;.
\]
For notational convenience, set
\begin{align*}
    \Pb^k_{t+1,\np_{1:t},w_{1:t}}
&\coloneqq
    \Pb_{(\nP_{t+1},W_{t+1}^k) \mid (\nP_1,\nW_1^k,\dots,\nP_{t},W_{t}^k) = (\np_1,\nw_1,\dots,\np_{t},w_{t}))}
\\
    \Pb_{t+1,\np_{1:t},w_{1:t}}
&\coloneqq
    \Pb_{(\nP_{t+1},W_{t+1}) \mid (\nP_1,\nW_1,\dots,\nP_{t},W_{t}) = (\np_1,\nw_1,\dots,\np_{t},w_{t}))}
\end{align*}
and notice that, for each $t \in [T-1]$ such that $t \ge 2$, using the chain rule for the Kullback-Leibler divergence \cite[Theorem 2.5.3]{ThomasM.Cover2006}, we have that
\begin{align*}
&
    \cD_{\text{KL}}\lrb{\Pb_{(\nP_1,\nW_1,\dots,\nP_t,\nW_t)},\Pb_{(\nP_1,\nW_1^k,\dots,\nP_t,\nW_t^k)}} \\
&=
    \cD_{\text{KL}}\lrb{\Pb_{(\nP_1,\nW_1,\dots,\nP_{t-1},W_{t-1})},\Pb_{(\nP_1,\nW_1^k,\dots,\nP_{t-1},W_{t-1}^k)}}
\\
&\qquad+
    \int_{ \lrb{[0,1] \times \{0,1\}^2}^{t-1}} \cD_{\text{KL}}
    \big(
        \Pb_{t,\np_{1:t-1},w_{1:t-1}},
        \Pb^k_{t,\np_{1:t-1},w_{1:t-1}}
    \big)
\\ & \qquad\qquad\qquad\qquad
    \dif \Pb_{(\nP_1,\nW_1,\dots,\nP_{t-1},W_{t-1})}(\np_1,\nw_1,\dots,\np_{t-1},w_{t-1})
\\
&\le
    \cD_{\text{KL}}\lrb{\Pb_{(\nP_1,\nW_1^k,\dots,\nP_{t-1},W_{t-1}^k)},\Pb_{(\nP_1,\nW_1,\dots,\nP_{t-1},W_{t-1})}}
\\
&\qquad+
    \nc_1 \cdot \e_K^2 \cdot \Pb\lsb{(\nS_t,\nB_t) \in R^{\textbf{exploit}}_k }
\\
&\qquad+
    c_2 \cdot \e_K \cdot \bigg(\Pb\lsb{(\nB_t,\nS_t) \in R^{\textbf{left}}_k}
    + \sum_{j=1}^{k-1} \Pb\lsb{(\nB_t,\nS_t) \in R^{\textbf{square}}_{k,j}}
\\
&\qquad\qquad +
    \Pb\lsb{(\nB_t,\nS_t) \in R^{\textbf{triangle}}_{k}}
    + \sum_{i=1}^{k-1} \Pb\lsb{(\nB_t,\nS_t) \in R^{\textbf{square}}_{i,k}}
    + \Pb\lsb{(\nB_t,\nS_t) \in R^{\textbf{top}}_k}\bigg)
\end{align*}
and iterating (and repeating essentially the same calculations in the last step where there is no conditioning), we get
\begin{align*}
&
\cD_{\text{KL}}\lrb{\Pb_{(\nP_1,\nW_1,\dots,\nP_t,\nW_t)},\Pb_{(\nP_1,\nW_1^k,\dots,\nP_t,\nW_t^k)}}
\\&\le
    \nc_1 \cdot \e_K^2 \cdot \E\lsb{ N^{\textbf{exploit}}_k(t-1)}
\\
&\qquad+
    c_2 \cdot \e_K \cdot \bigg(\E\lsb{N^{\textbf{left}}_k(t-1)}
    + \sum_{j=1}^{k-1} \E\lsb{N^{\textbf{square}}_{k,j}(t-1)}
\\
&\qquad\qquad+
    \E\lsb{N^{\textbf{triangle}}_{k}(t-1)}
    + \sum_{i=1}^{k-1} \E\lsb{N^{\textbf{square}}_{i,k}(t-1)}
    + \E\lsb{N^{\textbf{top}}_k(t-1)}\bigg)
\\
&\le
    \nc_1 \cdot \e_K^2 \cdot \E\lsb{ N^{\textbf{exploit}}_k(T)}
\\
&\qquad+
    c_2 \cdot \e_K \cdot \bigg(\E\lsb{N^{\textbf{left}}_k(T)}
    + \sum_{j=1}^{k-1} \E\lsb{N^{\textbf{square}}_{k,j}(T)}
\\
&\qquad\qquad+
    \E\lsb{N^{\textbf{triangle}}_{k}(T)}
    + \sum_{i=1}^{k-1} \E\lsb{N^{\textbf{square}}_{i,k}(T)}
    + \E\lsb{N^{\textbf{top}}_k(T)}\bigg)
\end{align*}
It follows that, for each $k \in [K]$,
\begin{align*}
    (\star_k)
&\le
    T \cdot \sqrt{ \frac{1}{2} } \cdot \bigg( \nc_1 \cdot \e_K^2 \cdot \E\lsb{ N^{\textbf{exploit}}_k(T)}
+
    c_2 \cdot \e_K \cdot \bigg(\E\lsb{N^{\textbf{left}}_k(T)}
    + \sum_{j=1}^{k-1} \E\lsb{N^{\textbf{square}}_{k,j}(T)}
\\
& \qquad
+
    \E\lsb{N^{\textbf{triangle}}_{k}(T)}
    + \sum_{i=1}^{k-1} \E\lsb{N^{\textbf{square}}_{i,k}(T)}
    + \E\lsb{N^{\textbf{top}}_k(T)}\bigg) \bigg)^{\nicefrac{1}{2}}
\\
&\le
      T \cdot \sqrt{ \frac{1}{2} } \cdot \e_K \cdot \sqrt{ \nc_1 \E\lsb{ N^{\textbf{exploit}}_k(T)}}
\\
&\qquad +
    T \cdot \sqrt{ \frac{1}{2} } \cdot \sqrt{\e_K} \cdot \bigg( c_2 \cdot \bigg(\E\lsb{N^{\textbf{left}}_k(T)}
    + \sum_{j=1}^{k-1} \E\lsb{N^{\textbf{square}}_{k,j}(T)}
\\
&\qquad\qquad
+
    \E\lsb{N^{\textbf{triangle}}_{k}(T)}
    + \sum_{i=1}^{k-1} \E\lsb{N^{\textbf{square}}_{i,k}(T)}
    + \E\lsb{N^{\textbf{top}}_k(T)}\bigg) \bigg)^{\nicefrac{1}{2}}
\end{align*}
For notational convenience, define
\[
    N^{\textbf{explore}}
\coloneqq
    \sum_{i \in [K]} N^{\textbf{left}}_i(T) + \sum_{j \in [K]} N^{\textbf{top}}_j(T) + \sum_{k \in [K]} N^{\textbf{triangle}}_k(T) +
    \sum_{i,j\in [K], \, i+j\le K} N^{\textbf{square}}_{i,j}(T)
\]
and
\[
    N^{\textbf{exploit}}
\coloneqq
    \sum_{k \in [K]} N^{\textbf{exploit}}_k(T)
\]
Notice that, after bringing the summation under the square root leveraging Jensen's inequality, we sum each of the terms $N^{\textbf{left}}_i(T), N^{\textbf{top}}_j(T), N^{\textbf{triangle}}_k(T), N^{\textbf{square}}_{i,j}(T)$ at most two times, we get:
\begin{align*}
    \frac{1}{K}\sum_{k \in [K]} (\star_k)
&\le
    T\cdot \e_K \cdot \sqrt{\frac{\nc_1}{2}} \cdot \sqrt{\frac{\E\bsb{N^{\textbf{exploit}}}}{K}}
+
    T\cdot \sqrt{\e_K} \cdot \sqrt{ c_2} \cdot \sqrt{ \frac{\E\bsb{N^{\textbf{explore}}}}{K}}
\\
&\le\;
    T\cdot \e_K \cdot \sqrt{\frac{\nc_1}{2}} \cdot \sqrt{\frac{T}{K}}
+
    T\cdot \sqrt{\e_K} \cdot \sqrt{ c_2} \cdot \sqrt{ \frac{\E\bsb{N^{\textbf{explore}}}}{K}} \;.
\end{align*}
It follows that
\begin{align*}
     (\square)
&=
    c_{\textbf{spike}} \cdot \e_K \cdot \lrb{ T - \frac{1}{K}\sum_{k \in [K]} \E[N^{\textbf{exploit}}_{k,k}]}
\\
&\ge
    c_{\textbf{spike}} \cdot \e_K \cdot \lrb{ T - \frac{1}{K}\sum_{k \in [K]} \E[N^{\textbf{exploit}}_{k}] - \frac{1}{K}\sum_{k\in[K]}(\star_k)}
\\
&=
    c_{\textbf{spike}} \cdot \e_K  \cdot \lrb{ T - \frac{\E[N^{\textbf{exploit}}]}{K} - \frac{1}{K}\sum_{k\in[K]}(\star_k)}
\\
&\ge
    c_{\textbf{spike}} \cdot \e_K \cdot T \cdot \lrb{ 1 - \frac{1}{K} -  \e_K \cdot \sqrt{\frac{\nc_1}{2}} \cdot \sqrt{\frac{T}{K}} -  \sqrt{\e_K} \cdot \sqrt{ c_2} \cdot \sqrt{ \frac{\E\bsb{N^{\textbf{explore}}}}{K}}}
\eqqcolon
    (\tilde{\square})\;.
\end{align*}
Now, if $\E\bsb{N^{\textbf{explore}}} \ge T^{\nicefrac{2}{3}}$ then
\[
    R_T \ge c_{\textbf{plat}} \cdot \E\bsb{N^{\textbf{explore}}} \ge c_{\textbf{plat}} \cdot T^{\nicefrac{2}{3}}\;.
\]
Instead, if $\E\bsb{N^{\textbf{explore}}}\le T^{\nicefrac{2}{3}}$, then
\begin{align*}
    (\tilde{\square})
&\ge
    c_{\textbf{spike}} \cdot \e_K \cdot T \cdot \lrb{ 1 - \frac{1}{K} -  \e_K \cdot \sqrt{\frac{\nc_1}{2}} \cdot \sqrt{\frac{T}{K}} - \sqrt{\e_K} \cdot \sqrt{ c_2} \cdot \sqrt{ \frac{T^{\nicefrac{2}{3}}}{K}}}
\\
&=
    c_{\textbf{spike}} \cdot \frac{T}{16 \bigl\lceil T^{\nicefrac{1}{3}} \bigr\rceil } 
    \lrb{ 1 - \frac{1}{\bigl\lceil T^{\nicefrac{1}{3}} \bigr\rceil } -  \frac{1}{16 \bigl\lceil T^{\nicefrac{1}{3}} \bigr\rceil } \cdot \sqrt{\frac{\nc_1}{2}} \cdot \sqrt{\frac{T}{\bigl\lceil T^{\nicefrac{1}{3}} \bigr\rceil }} - \sqrt{\frac{1}{16 \bigl\lceil T^{\nicefrac{1}{3}} \bigr\rceil }} \cdot \sqrt{ c_2} \cdot \sqrt{ \frac{T^{\nicefrac{2}{3}}}{\bigl\lceil T^{\nicefrac{1}{3}} \bigr\rceil }}}
\\
&\ge
    10^{-6}\cdot T^{\nicefrac{2}{3}}\;
\end{align*}
for $T\ge 42$, where in the last step we have plugged in the values of $c_1$ and $c_2$.
Since we proved that
\[
    \frac{1}{K} \sum_{k \in [K]} R^k_T \ge 10^{-6}\cdot T^{\nicefrac{2}{3}} \;,
\]
then, there exists an instance $k\in[K]$ such that the regret of the algorithm over $(\nP_t,\nZ_{\e_K,r_K^k,t})_{t \in [T]}$ is at least $T^{\nicefrac{2}{3}}$, concluding the result.
\end{proof}
}

The previous proof relies on Lemma \ref{lemma:kl-upperbounds}, for which we need the following two lemmas.

\begin{lemma}\label{lemma:lambdalipschitz}
    For any pair $\e, r \in (0, 1]^2$, $\Lambda_{r, \e}$ is $\frac{2}{\e}$-Lipschitz.
\end{lemma}
\begin{proof}
    Fix $0 \le \nb \le \ns \le 1$, we consider all possible cases with respect to the interval $\mathcal{R} = \mathcal{R}^- \cup \mathcal{R}^+ = [ r - \nicefrac{\e}{2}, r] \cup (r, r + \nicefrac{\e}{2}]$ to prove that
    \[
        | \Lambda(\ns) - \Lambda(\nb) | \le \frac{2}{\e} |\ns - \nb|
    \]
    \begin{itemize}
        \item If $\nb, \ns \notin \mathcal{R}$, then $\Lambda(\ns) = \Lambda(\nb) = 0$ and the property is trivially true. 
        \item If $\nb, \ns \in \mathcal{R}^-$ or $\nb, \ns \in \mathcal{R}^+$ then
        \[
            |\Lambda(\ns) - \Lambda(\nb)|
            = \left| 1 - \frac{2}{\e}(r - \ns) - 1 + \frac{2}{\e}(r - \nb) \right|
            = \frac{2}{\e} \left| \ns - \nb \right|
        \]
        \item If $\nb \in \mathcal{R}^-$ and $\ns \in \mathcal{R}^+$ then
        \[
            |\Lambda(\ns) - \Lambda(\nb)|
            = \left| 1 - \frac{2}{\e}(\ns - r) - 1 + \frac{2}{\e}(r - \nb) \right|
            = \frac{2}{\e} \left| \ns - \nb - 2r \right|
            \le \frac{2}{\e} \left| \ns - \nb \right|
        \]
        because $|s - b| > 2r$.
        \item If $\nb \in \mathcal{R}^+$ and $\ns \notin \mathcal{R}$ then
        \[
            |\Lambda(\ns) - \Lambda(\nb)|
            = \left| 1 - \frac{2}{\e}(\nb - r) \right|
            = \frac{2}{\e} \left| \frac{\e}{2} + r - \nb \right|
            \le \frac{2}{\e} | \ns - \nb |
        \]
        because $\nicefrac{\e}{2} + r \le \ns$. The same reasoning applies to the case $\nb \notin \mathcal{R}$ and $\ns \in \mathcal{R}^+$.
        \item If $\nb \in \mathcal{R}^-$ and $\ns \notin \mathcal{R}$ then
        \[
            |\Lambda(\ns) - \Lambda(\nb)|
            = \left| 1 - \frac{2}{\e}(r - \nb) \right|
            = \frac{2}{\e} \left| \frac{\e}{2} - r + \nb \right|
            = \frac{2}{\e} \left| r - \frac{\e}{2} - \nb \right|
            \le \frac{2}{\e} | \ns - \nb |
        \]
        because $r - \nicefrac{\e}{2} \le - \ns$. The same reasoning applies to the case $b \notin \mathcal{R}$ and $\ns \in \mathcal{R}^-$.
    \end{itemize}
\end{proof}

The following lemma, together with Lemma \ref{lemma:lambdalipschitz}, will be used in the proof of Lemma \ref{lemma:kl-upperbounds}.

\begin{lemma}\label{lemma:Lprimebound}
    For all pairs $(\nb, \ns) \in \cU \setminus [p_\text{exploit},1]^2$, it holds that
    \[
        F(\ns) - F(\nb) \ge \frac{1}{6} (\ns - \nb)
    \]
\end{lemma}
\begin{proof}
    Notice that, for any pair $(\nb,\ns) \in \cU\setminus [p_\text{exploit},1]^2$,
    \[
        F(\ns) - F(\nb) = \int_\nb^\ns f(x) \dif x \ge \min_{c \in [\nb,\ns]} f(c) (\ns - \nb)\;,
    \]

    The minimum value of $f$ is found when approaching $\nicefrac{3}{16}$ from the right, which yields $f(c) \rightarrow \nicefrac{1}{6}$, therefore $F(\ns) - F(\nb) \ge \nicefrac{1}{6}\cdot(\ns - \nb)$.
\end{proof}

We conclude this section by proving a key lemma used in the proof of \Cref{t:lower-bound-lip-iv}.

\begin{lemma}
\label{lemma:kl-upperbounds}
    For all $K\in\N$, $k\in [K]$, If $(\ns,\nb) \in R^{\textbf{exploit}}_k$, then we have
    \[
        \cD_{\text{KL}} \lrb{
            \Pb_{(\one{\nZ_{t+1} < \nb}, \one{\ns < \nZ_{t+1}}) },\Pb_{(\one{\nZ_{t+1}^k< \nb}, \one{\ns < \nZ_{t+1}^k}})
        }
        \le \frac{2}{81}\cdot \e_K^2 \; ,
    \]
    while if $(\ns,\nb) \in R^{\textbf{left}}_k \cup R^{\textbf{square}}_{k,j} \cup R^{\textbf{triangle}}_{k} \cup R^{\textbf{square}}_{i,j} \cup R^{\textbf{top}}_{k}$, then,
    \[
        \cD_{\text{KL}}
        \lrb{
            \Pb_{(\one{\nZ_{t+1}<b}, \one{\ns < \nZ_{t+1}}) },\Pb_{(\one{\nZ_{t+1}^k<\nb}, \one{\ns < \nZ_{t+1}^k}})
        } \le \frac{65}{9} \cdot \e_K \;.
    \]
\end{lemma}
\begin{proof}
Fix $K \in \mathbb{N}$, $\e_K = \nicefrac{1}{16K}$ and $r^k_K = \nicefrac{3}{16} + \e_K(k - \nicefrac{1}{2})$, the KL divergence can be decomposed as follows
\begin{align*}
&
    \cD_{\text{KL}}
    \lrb{
        \Pb_{(\I\{\nZ_{t+1}<\nb\}, \I\{\ns < \nZ_{t+1}\}) },\Pb_{(\I\{\nZ_{t+1}^k<\nb\}, \I\{\ns < \nZ_{t+1}^k\}})
    } \\
& =
    \log \lrb{
        \frac{F(\nb)}{F_{r^k_K, \e_K}(\nb)}
    } F(\nb)
    + \log \lrb{
        \frac{1 - F(\ns)}{1 - F_{r^k_K, \e_K}(\ns)}
    } (1 - F(\ns)) \\
& \qquad \qquad
    + \log \lrb{
        \frac{F(\ns) - F(\nb)}{F_{r^k_K, \e_K}(\ns) - F_{r^k_K, \e_K}(\nb)}
    } (F(\ns) - F(\nb)) \\
& =
    \log \lrb{
        \frac{F(\nb)}{F(\nb) + \frac{\e_K}{18} \Lambda_{r^k_K, \e_K}(\nb)}
    } F(\nb)
    + \log \lrb{
        \frac{1 - F(\ns)}{1 - F(\ns) - \frac{\e_K}{18} \Lambda_{r^k_K, \e_K}(\ns)}
    } (1 - F(\ns)) \\
& \qquad \qquad
    + \log \lrb{
        \frac{F(\ns) - F(\nb)}{F(\ns) - F(\nb) + \frac{\e_K}{18} (\Lambda_{r^k_K, \e_K}(\ns) - \Lambda_{r^k_K, \e_K}(\nb))}
    } (F(\ns) - F(\nb)) = (\star)
\end{align*}
For sake of brevity, we will write $\Lambda$ instead of $\Lambda_{r^k_K, \e_K}$ for the remainder of the proof. Note that the first term is non-positive for any $\nb \in [0, 1]$ , while the second term is non-negative for any $\ns \in [0, 1]$.
Furthermore if $\Lambda(\nb) = 0$, then $\log(\nicefrac{F(\nb)}{F(\nb)})F(\nb) = 0$, likewise if $\Lambda(\ns) = 0$ then $\log(\nicefrac{F(\ns)}{F(\ns)})F(\ns) = 0$, if both $\Lambda(\nb) = \Lambda(\ns) = 0$ then the whole expression is zero.

\paragraph{Exploitation region}
If $(\ns,\nb) \in R_k^\text{exploit}$, then $\Lambda(\ns) = 0$ because all perturbations happen within $p_\text{left}$ and $p_\text{right}$ whereas $\ns > p_\text{exploit}$,
If in addition $\Lambda(\nb) = 0$ then the whole expression is zero, otherwise
\begin{align*}
&
    (\star) \le
    \log \lrb{ \frac{F(\ns) - F(\nb)}{F(\ns) - F(\nb) - \frac{\e_K}{18}\Lambda(\nb)} } (F(\ns) - F(\nb))
    + \log \lrb{
        \frac{F(\nb)}{F(\nb) + \frac{\e_K}{18} \Lambda(\nb)}
    } F(\nb)
\\ &
    = \log\lrb{ \frac{1}{\lrb{ 1-\frac{\e_K\Lambda(\nb)}{18(F(\ns) - F(\nb)}}^{(F(\ns)-F(\nb))} }}
    + \log\lrb{ \frac{1}{\lrb{ 1+\frac{\e_K\Lambda(\nb)}{18F(\nb)}}^{F(\nb)} }}
\\ &
    = \log\lrb{
        \frac{1}{
            \lrb{
                1 - \frac{1}{\frac{18(F(\ns) - F(\nb))}{\e_K\Lambda(\nb)}}    
            }^{
                \frac{18(F(\ns)-F(\nb))}{\e_K \Lambda(\nb)}
                \cdot \frac{\e_K \Lambda(\nb)}{18}
            }
        }
    }
    + \log\lrb{ \frac{1}{\lrb{ 1+\frac{1}{\frac{18F(\nb)}{\e_K\Lambda(\nb)}}}^{\frac{18F(\nb)}{\e_K \Lambda(\nb)} \cdot \frac{\e_K \Lambda(\nb)}{18}} }}
    = (\circ)
\end{align*}
We know that for $x>1$ both $(1-1/x)^x$ and $(1+1/x)^x$ are monotonically increasing functions of $x$. Now, note that both $F(\ns) - F(\nb)$ and $F(\nb)$ are lower bounded by a constant. Indeed, $p_\text{left} \leq \nb \leq p_\text{right}$ and $\ns\geq p_\text{exploit}$, thus $F(\ns) - F(\nb)\geq F(p_\text{exploit}) - F(p_\text{right})$ and $F(\nb)\geq F(p_\text{left})$. Let $c = \min\{F(p_\text{exploit}) - F(p_\text{right}), F(p_\text{left})\} = \nicefrac{1}{6}$.
Therefore we have that both
$$
    \frac{18(F(\ns) - F(\nb))}{\e_K\Lambda(\nb)} \ge \frac{18c}{\e_K\Lambda(\nb)} > 1
    \qquad
    \text{ and }
    \qquad
    \frac{18F(\nb)}{\e_K\Lambda(\nb)} \ge \frac{18c}{\e_K\Lambda(\nb)} > 1
$$

Now, using the monotonicity of the denominators, we can write
\begin{multline*}
    (\circ) \le
    \log \lrb{ \frac{1}{\lrb{ 1 - \frac{\e_K\Lambda(\nb)}{18c}}^c} }
    + \log\lrb{ \frac{1}{\lrb{ 1 + \frac{\e_K\Lambda(\nb)}{18c} }^c }}
    = c\cdot\log\lrb{\frac{1}{1-\lrb{ \frac{\e_K\Lambda(\nb)}{18c}}^2}}
\\
    \leq c\cdot\frac{\lrb{ \frac{\e_K\Lambda(\nb)}{18c}}^2}{1-\lrb{ \frac{\e_K\Lambda(\nb)}{18c}}^2}
    \leq \frac{4}{3}\cdot c\cdot\lrb{ \frac{\e_K\Lambda(\nb)}{18c}}^2
    = \frac{4}{3}\cdot\frac{1}{c}\cdot\frac{1}{18^2}\cdot\e_K^2
\end{multline*}
Here, we first use the inequality that $\log x \leq x - 1$, then we use the facts that $\Lambda(\nb) \leq 1$ and $\e_K \leq 9c$ (loose bound). In practice $c = \nicefrac{1}{6}$ and $\e_K \le 9c = \nicefrac{3}{2}$ holds for any $K \in \mathbb{N}$.

\paragraph{Exploration regions}
If $(\ns,\nb) \in R^{\textbf{left}}_k \cup R^{\textbf{square}}_{k,j} \cup R^{\textbf{triangle}}_{k} \cup R^{\textbf{square}}_{i,k} \cup R^{\textbf{top}}_{k}$, then consider each term in $(\star)$ individually. 

The first term is upper bounded by zero for any $b \in [0, 1]$. 

If $\ns > p_\text{right}$, then $\Lambda(\ns) = 0$ and the second term is zero, otherwise the second term can be bounded as follows:
\begin{multline*}
    \log \lrb{ \frac{1 - F(\ns)}{1 - F(\ns) - \frac{\e_K}{18}\Lambda(\ns)} } (1 - F(\ns))
    \le \frac{\e_K}{18} \cdot \frac{\Lambda(\ns)}{1-F(\ns)-\frac{\e_K}{18}\Lambda(\ns)}(1-F(\ns))
    \\
    = \frac{\e_K}{18} \cdot \frac{\Lambda(\ns)}{1-\frac{\e_K}{18}\frac{\Lambda(\ns)}{1-F(\ns)}}
    \leq \frac{\e_K}{18} \cdot \frac{1}{1-\frac{\e_K}{18}\frac{1}{1-F(p_\text{right})}}
    \leq \frac{\e_K}{9}
\end{multline*}
where we used the inequality $\log x < x-1$ and the fact that $\frac{1}{1-\frac{\e_K}{18}F(p_\text{right})} \leq \nicefrac{1}{2}$, which is always true because $\e_K \leq 9(1-F(p_\text{right})) \le \nicefrac{81}{11}$ holds for any $K \in \mathbb{N}$. 

Finally, we just need to bound the third term. We consider the following cases.
\begin{itemize}
    \item  First consider the case in which $\ns - \nb \leq \e_K$, since $F$ is Lipschitz with constant $L$ we have $F(\ns) - F(\nb) \leq L\e_K$, thus
    \[
        \log\lrb{ \frac{F(\ns) - F(\nb)}{F(\ns) - F(\nb) + \frac{\e_K}{18}(\Lambda(\ns)-\Lambda(\nb))}} (F(\ns) - F(\nb))
        = \log\lrb{ \frac{1}{1+\frac{\e_K}{18}\cdot\frac{\Lambda(\ns)-\Lambda(\nb)}{F(\ns)-F(\nb)}}} L\e_K
    \]
    which is maximized when $\frac{\Lambda(\ns)-\Lambda(\nb)}{F(\ns)-F(\nb)}$ is minimized, by leveraging the Lipschitzness of $\Lambda$ from Lemma~\ref{lemma:lambdalipschitz} together with the lower bound on $F(\ns) - F(\nb)$ from Lemma~\ref{lemma:Lprimebound} we get
    \[
        \frac{\Lambda(\ns) - \Lambda(\nb)}{F(\ns) - F(\nb)}
        \ge
        \frac{- \nicefrac{2}{\e_K} (\ns - \nb)}{\nicefrac{1}{6} (\ns - \nb)}
        = - \frac{12}{\e_K}
    \]
    Substituting in the above, we get the upper bound:
    \begin{align*}
        \log \lrb{
            \frac{1}{1 - \frac{\e_K}{18} \frac{12}{\e_K}}
        } L\e_K
        < 2 L \e_K
        = \frac{64}{9}\cdot\e_K
    \end{align*}
    where $L = \nicefrac{32}{9}$ by the definition of $F$ in \Cref{eq:lowerbound-cdf}. For the case $\ns = \nb$, note that this term is zero because $\Lambda(\ns) - \Lambda(\nb) = 0$ and can therefore be ignored.

    \item Next consider the opposite $\ns - \nb > \e_K$. By Lemma~\ref{lemma:Lprimebound} we know that $F(\ns) - F(\nb) \ge \nicefrac{1}{6} (\ns - \nb) \ge \nicefrac{\e_K}{6}$. We use the inequality $\log x \leq x-1$ to get:
    \begin{align*}
    &
        \log\lrb{\frac{F(\ns)-F(\nb)}{F(\ns)-F(\nb)+\frac{\e_K}{18}(\Lambda(\ns)-\Lambda(\nb))}} (F(\ns) -F(\nb)) \\
    &
        \leq -\frac{\e_K}{18} \cdot \frac{\Lambda(\ns)-\Lambda(\nb)}{F(\ns)-F(\nb)+\frac{\e_K}{18}(\Lambda(\ns)-\Lambda(\nb))} (F(\ns)-F(\nb))
    \\ &
        \leq \frac{\e_K}{18} \cdot \frac{F(\ns)-F(\nb)}{F(\ns)-F(\nb)-\frac{\e_K}{18}}
        = \frac{\e_K}{18} \cdot \frac{1}{1-\frac{\e_K}{18}\frac{1}{F(\ns)-F(\nb)}}
    \\ &
        \leq \frac{\e_K}{18} \cdot \frac{1}{1-\frac{\e_K}{18}\frac{6}{\e_K}}
        = \frac{1}{12} \cdot \e_K
    \end{align*}
\end{itemize}
In conclusion, the result holds with $c_1 = \nicefrac{2}{81}$ and $c_2 = \nicefrac{1}{9} + \nicefrac{64}{9} = \nicefrac{65}{9}$.
\end{proof}

\section{Proof of Theorem \ref{t:minimal-lower-bound-iid}}
\label{s:liner-lower-bound-appe}

    Given that the sequence is $(\nP_t,\nZ_t)_{t \in [T]}$ i.i.d., without any loss of generality we can consider deterministic algorithms to post prices $(\nB_1, \nS_1), \dots, (\nB_T, \nS_T)$.
    Fix a time horizon $T \in \N$.
    We build an instance where, for each $t \in [T]$, the random variable $\nP_t$ takes values in $\{0,1\}$.
    In this case, notice that there are at most $8^{T}$ sequences of feedback
    $$
        (\nP_1, \one{ \nB_1 \ge \nZ_1 }, \one{ \nS_1 < \nZ_1 }),
        \dots,
        (\nP_{T-1}, \one{ \nB_{T-1} \ge \nZ_{T-1} }, \one{ \nS_{T-1} < \nZ_{T-1} })
    $$
    and, consequently, the set $\cP \subset [0,1]$ where the algorithm posts prices up to time contains at most $\sum_{t=1}^{T-1} 2 \cdot 8^{t}$ different prices.
    Let $\e \in (0, \nicefrac{1}{4})$.
    Given that $\cP$ is finite, we can select $\nicefrac{1}{2}-\e < c<d < \nicefrac{1}{2}+\e$ such that the closed interval $[c,d]$ does not contain any point in $\cP$.
    Now, select $(\nP_1,\nZ_1),\dots,(\nP_T,\nZ_T)$ as an i.i.d.\ sequence drawn uniformly in the set $\{(0,d),(1,c)\}$.
    Notice that, for any time $t \in [T]$, a maker that posts $(\nb, \ns) \in \cU$ such that $c < \nb = \ns < d$ gains (deterministically) at least $\nicefrac{1}{2} - \e$ regardless of the realization of $(\nP_t,\nZ_t)$ (if $(\nP_t,\nZ_t) = (0,d)$ then the maker always sells while if $(\nP_t,\nZ_t) = (1,c)$ the maker always buys).
    On the other hand, the pair of prices $(\nB_t,\nS_t)$ posted by the algorithm at any time $t \in [T]$ is such that one and only one of the following alternatives hold true:
    \begin{itemize}
        \itemm $\nB_t \le \nS_t < c$. In this case, the maker always sells to the taker, but with probability $\nicefrac{1}{2}$ the maker gains at most $\nicefrac{1}{2} + \e$ (when $(\nP_t,\nZ_t) = (0,d)$) and with probability $\nicefrac{1}{2}$ the maker loses at least $\nicefrac{1}{2} - \e$ (when $(\nP_t,\nZ_t) = (1,c)$).
        \itemm $\nB_t < c$ and $d < \nS_t$. In this case, then the maker never sells or buys (because the proposed selling price is too high and the proposed buying price is too low) and hence the maker gains zero. 
        \itemm $d < \nB_t \le \nS_t$. In this case, then the maker always buys from the taker, but with probability $\nicefrac{1}{2}$ the maker gains at most $\nicefrac{1}{2} + \e$ (when $(\nP_t,\nZ_t) = (1,c)$) and with probability $\nicefrac{1}{2}$ the maker loses at least $\nicefrac{1}{2} - \e$ (when $(\nP_t,\nZ_t) = (0,d)$). 
    \end{itemize}
    Overall, the algorithm gains in expectation at most $0$ when the underlying instance is $(\nP_t,\nZ_t)_{t \in [T]}$.
    Hence, for any fixed pair of prices $(\nb, \ns) \in \cU$ such that $c <\nb = \ns < d$, we have that
    \begin{multline*}
            R_T
        \ge
            \E\lsb{\sum_{t=1}^T \util(\nb,\ns,\nP_t,\nZ_t)} - \E \lsb{\sum_{t=1}^T \util(\nB_t,\nS_t,\nP_t,\nZ_t)}
        \\ \ge
            \lrb{\frac{1}{2}-\e}\cdot T - \lrb{ \lrb{\frac{1}{2} + \e} \cdot \frac{1}{2} - \lrb{\frac{1}{2} - \e} \cdot \frac{1}{2}} \cdot T
        =
            \lrb{\frac{1}{2} - 2\e} \cdot T \;.
    \end{multline*}
Setting $\delta \coloneqq 2 \e$ and noticing that $\e$ was chosen arbitrarily in the interval $(0, \nicefrac{1}{4})$, the conclusion follows.

\end{document}